\documentclass[a4paper,10pt]{article}

\usepackage{amsfonts}
\usepackage{amssymb}
\usepackage{amsmath}
\usepackage{amsthm}
\usepackage{stmaryrd}
\usepackage{mathrsfs}

\usepackage{graphicx}

\theoremstyle{definition}
\newtheorem{definition}{Definition}

\theoremstyle{plain}
\newtheorem{lemma}{Lemma}
\newtheorem{corollary}{Corollary}
\newtheorem{theorem}{Theorem}
\newtheorem{proposition}{Proposition}

\theoremstyle{remark}
\newtheorem*{remark}{Remark}

\begin{document}

\title{Approximating Acceptance Probabilities of CTMC-Paths on Multi-Clock Deterministic Timed Automata}

\author{Hongfei Fu\thanks{Supported by a CSC (China Scholarship Council) scholarship.}\\{\small Lehrstuhl f\"{u}r Informatik 2, RWTH Aachen University, Germany}}

\date{}


\maketitle

\begin{abstract}
We consider the problem of approximating the probability mass of the set of timed paths under a continuous-time Markov chain (CTMC) that are accepted by a deterministic timed automaton (DTA). As opposed to several existing works on this topic, we consider DTA with multiple clocks. Our key contribution is an algorithm to approximate these probabilities using finite difference methods. An error bound is provided which indicates the approximation error. The stepping stones towards this result include rigorous proofs for the measurability of the set of accepted paths and the integral-equation system characterizing the acceptance probability, and a differential characterization for the acceptance probability.
\end{abstract}

\section{Introduction}

Continuous-time Markov chains (CTMCs)~\cite{feller.william.aiptia} are one of the most prominent models for performance and dependability analysis of real-time stochastic systems. They are the semantical backbones of Markovian queueing networks, stochastic Petri nets and calculi for system biology and so forth. The desired behaviour of these systems is specified by various measures such as reachability with time information, timed logics such as CSL\cite{DBLP:journals/tse/BaierHHK03,DBLP:conf/icalp/ZhangJNH11}, mean response time, throughput, expected frequency of errors, and so forth.

Verification of continuous-time Markov chains has received much attention in recent years~\cite{DBLP:journals/cacm/BaierHHK10}. Many applicable results have been obtained on time-bounded reachability~\cite{DBLP:journals/tse/BaierHHK03,DBLP:conf/fsttcs/FearnleyRSZ11}, CSL model checking~\cite{DBLP:journals/tse/BaierHHK03,DBLP:conf/icalp/ZhangJNH11}, and so forth. In this paper, we focus on verifying CTMC against timed automata specification. In particular we consider approximating the probabilities of sets of CTMC-paths accepted by a deterministic timed automata (DTA)~\cite{DBLP:journals/tcs/AlurD94,DBLP:conf/hybrid/BrazdilKKKR11}. In general, DTA represents a wide class of linear real-time specifications. For example, we can describe time-bounded reachability probability ``to reach target set $G\subseteq S$ within time bound $T$ while avoiding unsafe states $U\subseteq S$'' $(G\cap U=\emptyset)$ by the single-clock DTA $\mathcal{A}_1$ (Fig.~\ref{fig:dta:a}), and the property ``to reach target set $G\subseteq S$ within time bound $T_1$ while successively remaining in unsafe states $U\subseteq S$ for at most $T_2$ time'' $(G\cap U=\emptyset)$ by the two-clock DTA $\mathcal{A}_2$ (Fig.~\ref{fig:dta:b}), both with initial configuration $(q_0,\vec{0})$. (We omit redundant locations that cannot reach the accepting state.)

\begin{figure}
\begin{minipage}[t]{0.45\linewidth}
\centering
\includegraphics{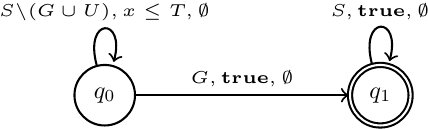}
\caption{DTA $\mathcal{A}_1$}\label{fig:dta:a}
\end{minipage}
\begin{minipage}[t]{0.55\linewidth}
\centering
\includegraphics{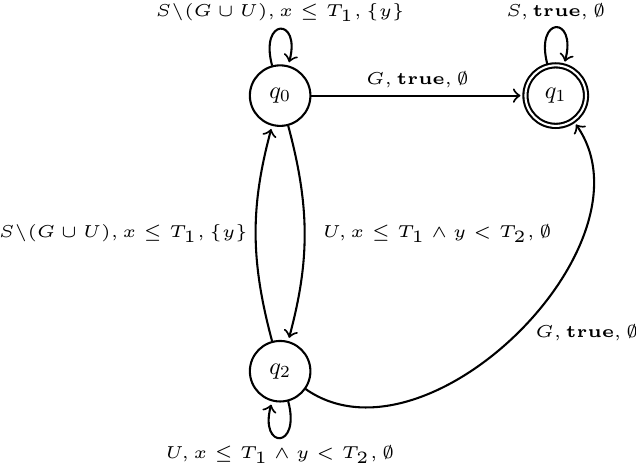}
\caption{DTA $\mathcal{A}_2$}\label{fig:dta:b}
\end{minipage}
\end{figure}

The problem to verify CTMC against DTA specifications is first considered by Donatelli~\emph{et al.}~\cite{DBLP:journals/tse/DonatelliHS09} where they enriched CSL with an acceptance condition of one-clock DTA to obtain the logic $\mathrm{CSL}^{\mathrm{TA}}$. In their paper, they proved that $\mathrm{CSL}^{\mathrm{TA}}$ is at least as expressive as $\mathrm{CSL}$ and $\mathrm{asCSL}$~\cite{DBLP:journals/tse/BaierHHK03,DBLP:journals/tse/BaierCHKS07}, and is strictly more expressive than $\mathrm{CSL}$. Moreover, they presented a model-checking algorithm for $\mathrm{CSL}^{\mathrm{TA}}$ using Markov regenerative processes. Chen~\emph{et al.}~\cite{DBLP:journals/corr/abs-1101-3694} systematically studied the DTA acceptance condition on CTMC-paths. More specifically, they proved that the set of CTMC-path accepted by a DTA is measurable and proposed a system of integral equations which characterizes the acceptance probabilities. Moreover, they demonstrated that the product of CTMC and DTA is a piecewise deterministic Markov process~\cite{mha.davis}, a dynamic system which integrates both discrete control and continuous evolution. Afterwards, Barbot~\emph{et al.}~\cite{DBLP:conf/tacas/BarbotCHKM11} put the approximation of DTA acceptance probabilities of CTMC-paths into practice, especially the algorithm on one-clock DTA which is first devised by Donatelli~\emph{et al.}~\cite{DBLP:journals/tse/DonatelliHS09} and then rearranged by Chen~\emph{et al.}~\cite{DBLP:journals/corr/abs-1101-3694}. Later on, Chen~\emph{et al.}~\cite{DBLP:conf/formats/ChenDKM11} proposed approximation algorithms for time-bounded verification of several linear real-time specifications, where the restricted time-bounded case, in which the time guard $x<T$ with a fresh clock $x$ and a time bound $T$ is enforced on each edge that leads to some final state of the DTA, is covered. Very recently, Mikeev~\emph{et al.}~\cite{mikeev.neuhausser.spieler.wolf} applies the notion of DTA acceptance condition on CTMC-paths to system biology. It is worth noting that Br\'{a}zdil~\emph{et al.} also studied DTA specifications in~\cite{DBLP:conf/hybrid/BrazdilKKKR11}. However they focused on semi-Markov processes as the underlying continuous-time stochastic model and limit frequencies of locations (in the DTA) as the performance measures, rather than path-acceptance.

Our contributions are as follows.  We start by providing a rigorous proof for the measurability of CTMC paths accepted by a DTA, correcting the proof provided by Chen \emph{et al.}~\cite{DBLP:journals/corr/abs-1101-3694}.  We confirm the correctness of the integral equation system characterizing acceptance probabilities provided by Chen \emph{et al.}~\cite{DBLP:journals/corr/abs-1101-3694} by providing a formal proof, and derive a differential characterization.  This provides the main basis for our algorithm to approximate acceptance probabilities using finite difference methods~\cite{jw.thomas}.  We provide tight error bounds for the approximation algorithm.  Whereas other works~\cite{DBLP:conf/tacas/BarbotCHKM11, mikeev.neuhausser.spieler.wolf, DBLP:journals/tse/DonatelliHS09} focus on single-clock DTA, our approximation scheme is applicable to any multi-clock DTA.  To our knowledge, this is the first such approximation algorithm with error bounds. Barbot \emph{et al.}~\cite{DBLP:conf/tacas/BarbotCHKM11} suggested an approximation scheme, but did not provide any error bounds.

The paper is organized as follows. In Section 2 we introduce some preliminaries. In Section 3 we prove the measurability of accepted paths, and prove the integral equations~\cite{DBLP:journals/corr/abs-1101-3694} that characterize the acceptance probability. In Section 4 we develop several tools useful to our main result. In Section 5 we propose a differential characterization for the family of acceptance probability functions. Base on these results, we establish and solve our approximation scheme in Section 6 by using finite difference methods~\cite{jw.thomas}, which is the main result of the paper. Section 7 concludes the paper and discusses some possible future works.

All integrals in this paper should be basically understood as Lebesgue Integral.

\section{Preliminaries}

In this section we introduce continuous-time Markov chains~\cite{feller.william.aiptia} and deterministic timed automata~\cite{DBLP:journals/tcs/AlurD94,DBLP:conf/hybrid/BrazdilKKKR11,DBLP:journals/corr/abs-1101-3694}.

\subsection{Continuous-Time Markov Chains}

\begin{definition}\label{def:ctmc}
A \emph{continuous-time Markov chain} (CTMC) is a tuple
$(S,L,\mathbf{P},\lambda,\mathcal{L})$ where
\begin{itemize}\itemsep0pt \parskip0pt \parsep0pt \topsep0pt
\item $S$ is a finite set of \emph{states}, and $L$ is a finite set of \emph{labels};
\item $\mathbf{P}:S\times S\mapsto [0,1]$ is a \emph{transition
matrix} such that $\sum_{u\in S}\mathbf{P}(s,u)=1$ for all $s\in S$;
\item $\lambda:S\mapsto\mathbb{R}_{>0}$ is an \emph{exit-rate function}, and $\mathcal{L}:S\mapsto L$ is a \emph{labelling function}.
\end{itemize}
\end{definition}
Intuitively, the running behaviour of a CTMC is as follows. Suppose $s$ is the current state of a CTMC. Firstly, the CTMC stays at $s$ for $t$ time units where the dwell-time $t$ observes the negative exponential distribution with rate $\lambda(s)$. Then the CTMC changes its current state to some state $u$ with probability $\mathbf{P}(s,u)$ and continues running from $u$, and so forth. The one-step probability of the transition from $s$ to $u$ whose dwell time lies in the interval $I$ equals $\mathbf{P}(s,u)\cdot\int_{t\in I}\lambda(s)\cdot e^{-\lambda(s)t}\,\mathrm{d}t$. Besides, the labelling function $\mathcal{L}$ assigns each state $s$
a label which indicates the set of atomic properties that hold at $s$.

To ease the notation, we denote the probability density function of the negative exponential distribution with rate $\lambda(s)$ by $\Lambda_s$, i.e., $\Lambda_s(t)=\lambda(s)\cdot e^{-\lambda(s)\cdot t}$ when $t\ge 0$ and $\Lambda_s(t)=0$ when $t<0$. 
It is worth noting that under our definition, we restrict ourselves such that the rates of all states are positive. CTMCs which contain states with rate $0$ (i.e. \emph{deadlock} states without outgoing transitions) can be adjusted to our case by (i) changing the rate of a deadlock state $s$ to any positive value and (ii) setting $P(s,s)=1$ and $P(s,u)=0$ for all $u\ne s$, i.e., by making a self-loop on $s$.

Below we formally define a probability measure on sets of CTMC-paths, following the definitions from~\cite{DBLP:journals/tse/BaierHHK03}.
Suppose $\mathcal{M}=(S,L,\mathbf{P},\lambda,\mathcal{L})$ be a CTMC.
An $\mathcal{M}$-\emph{path} $\pi$ is an infinite sequence
$s_0t_0s_1t_1\dots$ such that $s_n\in S$ and $t_n\in\mathbb{R}_{\ge 0}$
for all $n\in\mathbb{N}_0$. In other words, the set of $\mathcal{M}$-paths, denoted by $\mathrm{Path}(\mathcal{M})$, is essentially
$(S\times\mathbb{R}_{\ge 0})^\omega$. Given an $\mathcal{M}$-path $\pi=s_0t_0s_1t_1\dots$,
we denote $s_n$ and $t_n$ by $\pi[n]$ and $\pi{\langle}{n}{\rangle}$, respectively.

A \emph{template} $\theta$ is a finite sequence $s_0 I_0\dots s_{l-1}I_{l-1}s_{l}$
such that $l\ge 1$, $s_n\in S$ for all $0\le n\le l$ and $I_n$ is an interval in $\mathbb{R}_{\ge 0}$
for all $0\le n\le l-1$. Given a template $\theta=s_0 I_0\dots s_{l-1}I_{l-1}s_{l}$,
we define the \emph{cylinder set} $R_\theta$ as the following set:
\[
\{\pi\in\mathrm{Path}(\mathcal{M})\mid\pi[n]=s_n\mbox{ for all }0\le n\le l,\mbox{
and }\pi{\langle}{n}{\rangle}\in I_n\mbox{
for all }0\le n\le l-1\}
\]
An \emph{initial distribution} $\Theta$ is a function $S\mapsto [0,1]$ such that $\sum_{s\in S}\Theta(s)=1$. 
The \emph{probability space} $(\Omega,\mathcal{F},\mathcal{B}_\Theta)$
over $\mathcal{M}$-paths with initial distribution $\Theta$ is defined as follows:
\begin{itemize}\itemsep0pt \parskip0pt \parsep0pt \topsep0pt
\item $\Omega=\mathrm{Path}(\mathcal{M})$;
\item $\mathcal{F}\subseteq 2^{\Omega}$ is the smallest $\sigma$-algebra generated by the \emph{cylindrical family}
$\{R_\theta\mid\theta \mbox{ is a template}\}$ of subsets of $\Omega$.
\item $\mathcal{B}_\Theta:\mathcal{F}\mapsto [0,1]$ is the unique probability measure such that

\centerline{$\mathcal{B}_\Theta\left(R_\theta\right)=\Theta(s_0)\cdot\prod_{n=0}^{l-1}\left\{\mathbf{P}(s_n,s_{n+1})\cdot\int_{I_n}\Lambda_{s_n}(t)\,\mathrm{d}t\right\}$} for every
template $\theta=s_0 I_0\dots s_{l-1}I_{l-1}s_{l}$.
\end{itemize}

Intuitively, the probability space $(\Omega,\mathcal{F},\mathcal{B}_\Theta)$ is generated by all cylinder sets $R_\theta$, where $\mathcal{B}_\Theta(R_{\theta})$ is the product of the initial probability and those one-step probabilities specified in $\Theta$ and $\theta$.
The uniqueness of $\mathcal{B}_\Theta$ is guaranteed by Carath\'{e}odory's Extension Theorem~\cite{patrick.billingsley.pm}. 

When $\Theta$ is a Dirac distribution on $s\in S$ (i.e., $\Theta(s)=1$), we simply denote $(\Omega,\mathcal{F},\mathcal{B}_\Theta)$ by $(\Omega,\mathcal{F},\mathcal{B}_s)$. In this paper, we focus on the computation of $\mathcal{B}_s$, since any $\mathcal{B}_\Theta$ can be expressed as a linear combination of $\{\mathcal{B}_s\}_{s\in S}$.

\subsection{Deterministic Timed Automata}

Suppose $\mathcal{X}$ be a finite set of \emph{clocks}. A (\emph{clock}) \emph{valuation} over $\mathcal{X}$ is a function $\eta:\mathcal{X}\mapsto\mathbb{R}_{\ge 0}$. We denote by $\mathrm{Val}(\mathcal{X})$ the set of valuations over $\mathcal{X}$. Sometimes we will view a clock valuation as a real vector with an implicit order on $\mathcal{X}$.

A \emph{guard} (or \emph{clock constraints}) over a finite set of clocks $\mathcal{X}$ is a finite conjunction of basic constraints of the form $x\Join c$, where $x\in\mathcal{X}$, $\Join\in\{<,\le,>,\ge\}$ and $c\in\mathbb{N}_0$. We denote the set of guards over $\mathcal{X}$ by $\Phi(\mathcal{X})$. For each $\eta\in\mathrm{Val}(\mathcal{X})$ and $g\in\Phi(\mathcal{X})$, the satisfaction relation $\eta\models g$ is defined by: $\eta\models x\Join c$ iff $\eta(x)\Join c$, and $\eta\models g_1\wedge g_2$ iff $\eta\models g_1$ and $\eta\models g_2$. Given $g\in\Phi(\mathcal{X})$, we may also refer $g$ to the set of valuations that satisfy $g$: this may happen in the context such as $g_1\cap g_2$, etc. Given $X\subseteq\mathcal{X}$, $\eta\in\mathrm{Val}(\mathcal{X})$ and $t\in\mathbb{R}_{\ge 0}$, the valuations $\eta[X:=0]$, $\eta+t$, and $\eta-t$ are defined as follows:
\begin{enumerate}\itemsep0pt \parskip0pt \parsep0pt \topsep0pt
\item if $x\in X$ then $\eta[X:=0](x):=0$, otherwise $\eta[X:=0](x):=\eta(x)$\enskip;
\item $(\eta+t)(x):=\eta(x)+t$ for all $x\in\mathcal{X}$\enskip;
\item $(\eta-t)(x):=\eta(x)-t$ for all $x\in\mathcal{X}$, provided that $\eta(x)\ge t$ for all $x\in\mathcal{X}$\enskip.
\end{enumerate}
Intuitively, $\eta[X:=0]$ is obtained by resetting all clocks of $X$ to zero on $\eta$, and $\eta+t$ resp. $\eta-t$ is obtained by delaying resp. backtracking $t$ time units from $\eta$.

\begin{definition}\cite{DBLP:journals/tcs/AlurD94,DBLP:journals/corr/abs-1101-3694,DBLP:conf/hybrid/BrazdilKKKR11}
A \emph{deterministic timed automaton} (DTA) is a tuple \\$(Q,\Sigma,\mathcal{X},\Delta,F)$ where
\begin{itemize} \itemsep0pt \parskip0pt \parsep0pt \topsep0pt
\item $Q$ is a finite set of \emph{locations}, and  $F\subseteq Q$ is a set of \emph{final locations};
\item $\Sigma$ is a finite \emph{alphabet} of \emph{signatures}, and $\mathcal{X}$ is a finite set of \emph{clocks};
\item $\Delta\subseteq Q\times\Sigma\times\Phi(\mathcal{X})\times 2^\mathcal{X}\times Q$ is a finite set of \emph{rules} such that
\begin{enumerate} \itemsep0pt \parskip0pt \parsep0pt \topsep0pt
\item $\Delta$ is \emph{deterministic}: whenever $(q_1,a_1,g_1,X_1,q'_1),(q_2,a_2,g_2,X_2,q'_2)\in\Delta$, if $(q_1,a_1)=(q_2,a_2)$ and $g_1\cap g_2\ne\emptyset$ then
$(g_1,X_1,q'_1)=(g_2,X_2,q'_2)$.
\item $\Delta$ is \emph{total}: for all $(q,a)\in Q\times\Sigma$ and $\eta\in\mathrm{Val}(\mathcal{X})$, there exists $(q,a,g,X,q')\in\Delta$ such that $\eta\models g$.
\end{enumerate}
\end{itemize}
Given $q\in Q$, $\eta\in\mathrm{Val}(\mathcal{X})$ and $a\in\Sigma$, the triple $(\mathbf{g}_{q,a}^{\eta},\mathbf{X}_{q,a}^{\eta},\mathbf{q}_{q,a}^{\eta})\in\Phi(\mathcal{X})\times 2^\mathcal{X}\times Q$ are determined such that $(q,a,\mathbf{g}_{q,a}^{\eta},\mathbf{X}_{q,a}^{\eta},\mathbf{q}_{q,a}^{\eta})\in\Delta$ is the unique rule satisfying $\eta\models \mathbf{g}_{q,a}^{\eta}$.
\end{definition}

\begin{definition}\cite{DBLP:journals/tcs/AlurD94,DBLP:journals/corr/abs-1101-3694}
Let $\mathcal{A}=(Q,\Sigma,\mathcal{X},\Delta,F)$ be a DTA. A \emph{configuration} of $\mathcal{A}$ is a pair $(q,\eta)$, where $q\in Q$ and $\eta\in\mathrm{Val}(\mathcal{X})$. A \emph{timed signature} is a pair $(a,t)$ where $a\in\Sigma$ and $t\in\mathbb{R}_{\ge 0}$. The \emph{one-step transition function}
\[
\kappa^\mathcal{A}:(Q\times\mathrm{Val}(\mathcal{X}))\times (\Sigma\times\mathbb{R}_{\ge 0})\mapsto Q\times\mathrm{Val}(\mathcal{X})
\]
is defined by: $\kappa^\mathcal{A}((q,\eta),(a,t))=(\mathbf{q}_{q,a}^{\eta+t},(\eta+t)[\mathbf{X}_{q,a}^{\eta+t}:=0])$\enskip.
\end{definition}
We may represent $\kappa^\mathcal{A}((q,\eta),(a,t))=(q',\eta')$ by the more intuitive notation ``$(q,\eta)\xrightarrow{(a,t)}_{\mathcal{A}}(q',\eta')$''. We omit ``$\mathcal{A}$'' if the context is clear.

Intuitively, the configuration $\kappa((q,\eta),(a,t))$ is obtained as follows: firstly we delay $t$ time-units at $(q,\eta)$ to obtain $(q,\eta+t)$; then we find the unique rule $(q,a,g,X,q')\in\Delta$ such that $\eta+t\models g$; finally, we obtain $\kappa((q,\eta),(a,t))$ by changing the location to $q'$ and resetting $\eta+t$ with $X$. The determinism and the totality of $\Delta$ together ensures that $\kappa$ is a function.

\begin{definition}\cite{DBLP:journals/corr/abs-1101-3694}
Let $\mathcal{A}=(Q,\Sigma,\mathcal{X},\Delta,F)$ be a DTA. A \emph{timed word} is an infinite sequence of timed signatures. The \emph{run} of $\mathcal{A}$ on a timed word $w=\{(a_n,t_n)\}_{n\in\mathbb{N}_0}$ with \emph{initial configuration} $(q,\eta)$, denoted by $\mathcal{A}_{q,\eta}(w)$, is the unique infinite
sequence $\{(q_n,\eta_n)(a_n,t_n)\}_{n\in\mathbb{N}_0}$ which satisfies that $(q_0,\eta_0)=(q,\eta)$ and $(q_{n+1},\eta_{n+1})=\kappa^\mathcal{A}((q_n,\eta_n),(a_n,t_n))$ for $n\ge 0$.

A timed word $w$ is \emph{accepted by} $\mathcal{A}$ with initial configuration $(q,\eta)$ (\emph{abbr}. ``$w$ accepted by $\mathcal{A}_{q,\eta}$'') iff $\mathcal{A}_{q,\eta}(w)=\{(q_n,\eta_n)(a_n,t_n)\}_{n\in\mathbb{N}_0}$ satisfies that $q_n\in F$ for some $n\ge 0$. Moreover, $w$ is \emph{accepted} by $\mathcal{A}_{q,\eta}$ \emph{within} $k$
\emph{steps} ($k\ge 0$) iff $\mathcal{A}_{q,\eta}(w)=\{(q_n,\eta_n)(a_n,t_n)\}_{n\in\mathbb{N}_0}$
satisfies that $q_n\in F$ for some $0\le n\le k$.
\end{definition}

\section{Measurability and The Integral Equations}

In this section, we provide a rigorous proof for the measurability of the set of CTMC-paths accepted by a DTA and the system of integral equations that characterizes the acceptance probability. The notion of acceptance follows the previous ones in~\cite{DBLP:journals/corr/abs-1101-3694}.

Below we fix a CTMC $\mathcal{M}=(S,L,\mathbf{P},\lambda,\mathcal{L})$ and a DTA $\mathcal{A}=(Q,\Sigma,\mathcal{X},\Delta,F)$ such that $\Sigma=L$. Given a finite or infinite word $\alpha$, we denote by $\alpha_n$ ($n\ge 0$) the $n$-th signature, i.e., $\alpha=\alpha_0\alpha_1\dots\alpha_n\dots$ if $\alpha$ is infinite and $\alpha=\alpha_0\alpha_1\dots\alpha_{k-1}$ if $\alpha$ is finite with length $k$. Analogously, given a $k$-dimensional vector $\mathbf{t}$, we denote $\mathbf{t}=(\mathbf{t}_0,\dots,\mathbf{t}_{k-1})$.
We denote by ${\langle}{J}{\rangle}$ the characteristic function of a set $J$.

Firstly, we formally define the notion of acceptance.

\begin{definition}\cite{DBLP:journals/corr/abs-1101-3694}
The set of $\mathcal{M}$-paths \emph{accepted} by $\mathcal{A}$ w.r.t $s\in S$, $q\in Q$ and $\eta\in\mathrm{Val}(\mathcal{X})$, denoted by $\mathsf{Path}^{\mathcal{M}\otimes\mathcal{A}}(s,q,\eta)$, is defined by:
\[
\mathsf{Path}^{\mathcal{M}\otimes\mathcal{A}}(s,q,\eta):=\{\pi\in
\mathrm{Path}(\mathcal{M})\mid\pi[0]=s\mbox{ and
}\mathcal{L}_\pi\mbox{ is accepted by }\mathcal{A}_{q,\eta}\}~~,
\]
where $\mathcal{L}_\pi$ is the timed word defined by: $(\mathcal{L}_\pi)_n=\left(\mathcal{L}(\pi[n]),\pi{\langle}{n}{\rangle}\right)$ for all $n\ge 0$. Moreover, the set of $\mathcal{M}$-paths accepted by $\mathcal{A}$ w.r.t $s$, $q$ and $\eta$ \emph{within} $k$-\emph{steps} ($k\ge 0$), denoted by $\mathsf{Path}^{\mathcal{M}\otimes\mathcal{A}}_k(s,q,\eta)$, is defined as the set of $\mathcal{M}$-paths $\pi$ such that $\pi[0]=s$ and $\mathcal{L}_\pi$ is accepted by $\mathcal{A}_{q,\eta}$
within $k$ steps.
\end{definition}
Note that $\mathcal{L}_\pi$ specifies the behaviour of $\mathcal{M}$ observable by an outside observer. By definition, we have $\bigcup_{k\ge 0}\mathsf{Path}_k^{\mathcal{M}\otimes\mathcal{A}}(s,q,\eta)=\mathsf{Path}^{\mathcal{M}\otimes\mathcal{A}}(s,q,\eta)$. We omit ``$\mathcal{M}\otimes\mathcal{A}$'' in
 if the underlying context is clear.

\begin{remark}
We point out the main error in the measurability proof by Chen~\emph{et al.}~\cite{DBLP:journals/corr/abs-1101-3694}. The error appears on Page 11 under the label ``(1b)'' which handles the equality guards in timed transitions. In (1b), for an timed transition $e$ emitted from $q$ with guard $x=K$, four DTA $\mathcal{A}_e,\overline{\mathcal{A}}_e$, $\mathcal{A}^>_e,\mathcal{A}^<_e$ are defined w.r.t the original DTA $\mathcal{A}$. Then it is argued that
\[
Paths^\mathcal{C}(\mathcal{A}_e)=Paths^\mathcal{C}(\overline{\mathcal{A}}_e)\backslash(Paths^\mathcal{C}(\mathcal{A}^>_e)\cup Paths^\mathcal{C}(\mathcal{A}^<_e))
\]
This is incorrect. The left part $Paths^\mathcal{C}(\mathcal{A}_e)$ excludes all timed paths which involve both the guard $x>K$ and the guard $x<K$ (from $q$). However the right part does not. So the left and right part are not equal. \qed
\end{remark}
Below we prove that $\mathsf{Path}(s,q,\eta)$ is measurable under $\Omega$ for $(s,q,\eta)\in S\times Q\times\mathrm{Val}(\mathcal{X})$, and the integral-equation system that characterizes the acceptance probability~\cite{DBLP:journals/corr/abs-1101-3694}.  We abbreviate  $\left(\mathbf{g}_{q,\mathcal{L}(s)}^{\eta},\mathbf{X}_{q,\mathcal{L}(s)}^{\eta},\mathbf{q}_{q,\mathcal{L}(s)}^{\eta}\right)$ as $(\mathbf{g}_{q,s}^{\eta},\mathbf{X}_{q,s}^{\eta},\mathbf{q}_{q,s}^{\eta})$.
Given $\gamma\in\Delta$, we may denote $\gamma=\left(\mathfrak{q}(\gamma),\mathfrak{a}(\gamma),\mathfrak{g}(\gamma),\mathfrak{X}(\gamma),\mathfrak{q}'(\gamma)\right)$\enskip.

Note that $\bigcup_{k\ge 0}\mathsf{Path}_k(s,q,\eta)=\mathsf{Path}(s,q,\eta)$. Thus in order to prove the measurability of $\mathsf{Path}(s,q,\eta)$, it suffices to prove that each $\mathsf{Path}_k(s,q,\eta)$ is measurable under $\Omega$. To this end, we decompose $\mathsf{Path}_k(s,q,\eta)$ into subsets of paths, as follows.

\begin{definition}\label{def:dcmp:1}
Suppose $k\in\mathbb{N}$. We define the set $D^k\subseteq S^{k+1}\times\Delta^{k}$ as follows. For all $(\alpha,\beta)\in S^{k+1}\times\Delta^k$, $(\alpha,\beta)\in D^k$ iff the following conditions hold:
\begin{itemize}\itemsep0pt \parskip0pt \parsep0pt \topsep0pt
\item $\mathfrak{a}(\beta_n)=\mathcal{L}(\alpha_n)$ for all $0\le n\le k-1$, and $\mathfrak{q}'(\beta_n)=\mathfrak{q}(\beta_{n+1})$ for all $0\le n\le k-2$;
\item either $\mathfrak{q}(\beta_n)\in F$ for some $0\le n\le k-1$, or $\mathfrak{q}'(\beta_{k-1})\in F$.
\end{itemize}
Let $D^k(q,s):=\{\left(\alpha,\beta\right)\in D^k\mid(\mathfrak{q}(\beta_0),\alpha_0)=(q,s)\}$, for each $q\in Q$ and $s\in S$.
\end{definition}

\begin{definition}\label{def:dcmp:2}
Suppose $(\alpha,\beta)\in D^k$ $(k\ge 1)$ and $\eta\in\mathrm{Val}(\mathcal{X})$. Define the set $\mathsf{Path}^\alpha_{\beta,\eta}\subseteq\mathrm{Path}(\mathcal{M})$ as follows. For all $\pi\in\mathrm{Path}(\mathcal{M})$, $\pi\in\mathsf{Path}^\alpha_{\beta,\eta}$ iff the following conditions hold:
\begin{itemize} \itemsep0pt \parskip0pt \parsep0pt \topsep0pt
\item $\pi[n]=\alpha_n$ for all $0\le n\le k$;
\item The run $\mathcal{A}_{\mathfrak{q}(\beta_0),\eta}(\mathcal{L}_\pi)=\{(q_n,\eta_n)(\mathcal{L}(\pi[n]),\pi{\langle}{n}{\rangle})\}_{n\ge 0}$ satisfies that $q_n=\mathfrak{q}(\beta_n)$ and $\eta_n+\pi{\langle}{n}{\rangle}\models\mathfrak{g}(\beta_n)$ for all $0\le n\le k-1$.
\end{itemize}
\end{definition}
The intuition is that $\mathsf{Path}^\alpha_{\beta,\eta}$ is the set of $\mathcal{M}$-paths which visit the first $k+1$ states in the state sequence $\alpha$ while $\mathcal{A}$ synchronizes with the timed path by taking $k$ rules from the sequence $\beta$ (cf.~Fig.~\ref{fig:decomp}). From Definition~\ref{def:dcmp:1}, Definition~\ref{def:dcmp:2} and the fact that $\mathcal{A}$ is deterministic, it is not hard to prove the following lemma.
\begin{figure}[h]
\[
\binom{\mathfrak{q}(\beta_0)}{\eta}\dots
\binom{\mathfrak{q}(\beta_{n})}{\eta_{n}}\xrightarrow[\beta_{n}]{\binom{\mathcal{L}(\alpha_{n})}{{\pi}{\langle}{n}{\rangle}}}
\binom{\mathfrak{q}(\beta_{n+1})}{\eta_{n+1}}\dots
\binom{\mathfrak{q}(\beta_{k-1})}{\eta_{k-1}}\xrightarrow[\beta_{k-1}]{\binom{\mathcal{L}(\alpha_{k-1})}{{\pi}{\langle}{k-1}{\rangle}}}
\binom{\mathfrak{q}'(\beta_{k-1})}{\eta_{k}}
\xrightarrow{\binom{\mathcal{L}(\alpha_{k})}{{\pi}{\langle}{k}{\rangle}}}\dots
\]
\caption{The run $\mathcal{A}_{\mathfrak{q}(\beta_0),\eta}(\mathcal{L}_\pi)$ given $\pi\in\mathsf{Path}^\alpha_{\beta,\eta}$}\label{fig:decomp}
\end{figure}

\begin{lemma}\label{lem:mrbl:dcmp}
For all $k\ge 1$, $\mathsf{Path}_k(s,q,\eta)=\bigcup\{\mathsf{Path}^\alpha_{\beta,\eta}\mid
(\alpha,\beta)\in D^k(q,s)\}$. Furthermore, the union is disjoint, i.e., $\mathsf{Path}^{\alpha}_{\beta,\eta}\cap\mathsf{Path}^{\alpha'}_{\beta',\eta}=\emptyset$ whenever $(\alpha,\beta)\ne(\alpha',\beta')$.
\end{lemma}
Thus to prove that $\mathsf{Path}_k(s,q,\eta)$ is measurable, it suffices to prove that each $\mathsf{Path}^\alpha_{\beta,\eta}$ is measurable. To this end, we prove two technical lemmas as follows.

\begin{lemma}\label{lem:dcmp:char}
Let $k\ge 1$. For each $\beta\in\Delta^k$ and $\eta\in\mathrm{Val}(\mathcal{X})$, we define
\[
J_{\beta,\eta}:=\{\mathbf{t}\in\mathbb{R}^k_{\ge 0}\mid\nu[\beta,\eta,\mathbf{t}]_n+\mathbf{t}_n\models\mathfrak{g}(\beta_n)\mbox{
for all }0\le n\le k-1\}
\]
where $\nu[\beta,\eta,\mathbf{t}]_n\in\mathrm{Val}(\mathcal{X})$ \emph{(}$0\le n\le k-1$\emph{)} is defined by: (cf. Fig.~\ref{fig:decomp:2})
\[
\nu[\beta,\eta,\mathbf{t}]_0=\eta; \quad \nu[\beta,\eta,\mathbf{t}]_{n+1}=(\nu[\beta,\eta,\mathbf{t}]_n+\mathbf{t}_n)[\mathfrak{X}(\beta_n):=0]\enskip.
\]
Then given any $(\alpha,\beta)\in D^k$ and $\eta\in\mathrm{Val}(\mathcal{X})$, $\pi\in\mathsf{Path}^{\alpha}_{\beta,\eta}$ iff $\pi[n]=\alpha_n$ for all $0\le n\le k$ and $(\pi{\langle}{0}{\rangle},\dots,\pi{\langle}{k-1}{\rangle})\in J_{\beta,\eta}$.
\end{lemma}
\begin{proof}
Suppose $\pi\in\mathsf{Path}^{\alpha}_{\beta,\eta}$. Let
$\mathcal{A}_{\mathfrak{q}(\beta_0),\eta}(\mathcal{L}_\pi)=\{(q_n,\eta_n)(\mathcal{L}(\pi[n]),\pi{\langle}{n}{\rangle})\}_{n\ge 0}$ and $\mathbf{t}:=(\pi{\langle}{0}{\rangle},\dots,\pi{\langle}{k-1}{\rangle})$. By $\pi\in\mathsf{Path}^{\alpha}_{\beta,\eta}$, $\pi[n]=\alpha_n$ for all $0\le n\le k$, and $q_n=\mathfrak{q}(\beta_n)$ and $\eta_n+\pi{\langle}{n}{\rangle}\models\mathfrak{g}(\beta_n)$ for all $0\le n\le k-1$. Then one can prove inductively on $n$ that $\eta_n=\nu[\beta,\eta,\mathbf{t}]_n$ for all $0\le n\le k-1$. Thus $\nu[\beta,\eta,\mathbf{t}]_n+\mathbf{t}_n\models\mathfrak{g}(\beta_n)$ for all $0\le n\le k-1$. It follows that  $\mathbf{t}\in J_{\beta,\eta}$.

Suppose now that $\pi[n]=\alpha_n$ for all $0\le n\le k$ and $\mathbf{t}:=(\pi{\langle}{0}{\rangle},\dots,\pi{\langle}{k-1}{\rangle})\in J_{\beta,\eta}$. Denote $\mathcal{A}_{\mathfrak{q}(\beta_0),\eta}(\mathcal{L}_\pi)=\{(q_n,\eta_n)(\mathcal{L}(\pi[n]),\pi{\langle}{n}{\rangle})\}_{n\ge 0}$. Since $\mathcal{A}$ is deterministic, one can prove inductively on $n$ that $q_n=\mathfrak{q}(\beta_n)$ and $\eta_n=\nu[\beta,\eta,\mathbf{t}]_n$ for all $0\le n\le k-1$. Then we have $\eta_n+\pi{\langle}{n}{\rangle}\models\mathfrak{g}(\beta_n)$ for all $0\le n\le k-1$.
\end{proof}
\begin{figure}[h]
\[
\eta=\nu[\beta,\eta,\mathbf{t}]_0\xrightarrow[\beta_0]{\mathbf{t}_0}\cdots\nu[\beta,\eta,\mathbf{t}]_n\xrightarrow[\beta_n]{\mathbf{t}_n}\nu[\beta,\eta,\mathbf{t}]_{n+1}
\dots\nu[\beta,\eta,\mathbf{t}]_{k-1}\xrightarrow[\beta_{k-1}]{\mathbf{t}_{k-1}}\bot
\]
\caption{The definition of $\nu[\beta,\eta,\mathbf{t}]_n$}\label{fig:decomp:2}
\end{figure}

\begin{remark}\label{rmk:1}
One can prove inductively that for all $0\le n\le k-1$ and for all $x\in\mathcal{X}$:
\begin{itemize}\itemsep0pt \parskip0pt \parsep0pt \topsep0pt
\item $\nu[\beta,\eta,\mathbf{t}]_n(x)+\mathbf{t}_n=\eta(x)+\sum_{i=0}^{n}\mathbf{t}_i$ if
$x\not\in\bigcup_{i=0}^{n-1}\mathfrak{X}(\beta_i)$; and
\item $\nu[\beta,\eta,\mathbf{t}]_n(x)+\mathbf{t}_n=\sum_{i=m+1}^{n}\mathbf{t}_i$ if $x\in\mathfrak{X}(\beta_m)\backslash\left(\bigcup_{i=m+1}^{n-1}\mathfrak{X}(\beta_i)\right)$ for some unique $m<n$.
\end{itemize}
Thus each $\nu[\beta,\eta,\mathbf{t}]_n(x)+\mathbf{t}_n$ is the summation of a possible constant and a consecutive segment of $\mathbf{t}_0,\dots,\mathbf{t}_{k-1}$.\qed
\end{remark}

\begin{lemma}\label{lem:cons:recr}
Let $k\ge 2$. Suppose $\beta\in\Delta^{k}$ and $\eta\in\mathrm{Val}(\mathcal{X})$. For all $\mathbf{t}\in\mathbb{R}^k_{\ge 0}$,
\[
\mathbf{t}\in J_{\beta,\eta}\mbox{ iff }\eta+\mathbf{t}_0\models\mathfrak{g}(\beta_0)\mbox{ and }\hat{\mathbf{t}}\in J_{\hat{\beta},(\eta+\mathbf{t}_0)[\mathfrak{X}(\beta_0):=0]}\enskip,
\]
where $\hat{\beta}=\beta_1\dots\beta_{k-1}$ and $\hat{\mathbf{t}}=(\mathbf{t}_1,\dots,\mathbf{t}_{k-1})$.
\end{lemma}
\begin{proof}
Suppose $\mathbf{t}\in\mathbb{R}^k_{\ge 0}$. Note that for all $0\le n\le k-2$\enskip, we have
\[
\nu[\beta,\eta,\mathbf{t}]_{n+1}=\nu\left[\hat{\beta},\nu[\beta,\eta,\mathbf{t}]_1,\hat{\mathbf{t}}\right]_n=\nu\left[\hat{\beta},(\eta+\mathbf{t}_0)[\mathfrak{X}(\beta_0):=0],\hat{\mathbf{t}}\right]_n\enskip.
\]
Then we obtain: $\mathbf{t}\in J_{\beta,\eta}$

\begin{tabular}{rl}
iff & $\nu[\beta,\eta,\mathbf{t}]_n+\mathbf{t}_n\models\mathfrak{g}(\beta_n)$ for all $0\le n\le k-1$ \\
iff & $\eta+\mathbf{t}_0\models\mathfrak{g}(\beta_0)$ and $\nu[\beta,\eta,\mathbf{t}]_n+\mathbf{t}_n\models\mathfrak{g}(\beta_n)$ for all $1\le n\le k-1$\\
iff & $\eta+\mathbf{t}_0\models\mathfrak{g}(\beta_0)$ and $\nu
\left[\hat{\beta},(\eta+\mathbf{t}_0)[\mathfrak{X}(\beta_0):=0],\hat{\mathbf{t}}\right]_n+\hat{\mathbf{t}}_n\models\mathfrak{g}({\hat{\beta}}_n)$\\
    & for all $0\le n\le k-2$\\
iff & $\eta+\mathbf{t}_0\models\mathfrak{g}(\beta_0)$ and $\hat{\mathbf{t}}\in J_{\hat{\beta},(\eta+\mathbf{t}_0)[\mathfrak{X}(\beta_0):=0]}$\enskip.\\
\end{tabular}

\end{proof}
Now we prove the measurability result and the integral equations~\cite{DBLP:journals/corr/abs-1101-3694}. First we demonstrate that each closed subset of $\mathbb{R}_{\ge 0}^k$ is measurable when equipped with some $\alpha\in S^{k+1}$. Below given $\alpha\in S^{k+1}$ and $W\subseteq\mathbb{R}_{\ge 0}^k$ with $k\ge 1$, we define $\mathsf{Path}[\alpha,W]$ as the following set:
\[
\{\pi\in\mathrm{Path}(\mathcal{M})\mid\pi[n]=\alpha_n\mbox{ for all }0\le n\le k\mbox{ and }(\pi{\langle}{0}{\rangle},\dots,\pi{\langle}{k-1}{\rangle})\in W\}\enskip.
\]

\begin{lemma}\label{lem:mrbl:closed}
Suppose $\alpha\in S^{k+1}$ and $W\subseteq\mathbb{R}_{\ge 0}^k$ with $k\ge 1$. If $W$ is closed, then $\mathsf{Path}[\alpha,W]$ is measurable under $\Omega$. Furthermore, the probability mass of $\mathsf{Path}[\alpha, W]$ under $\mathcal{B}_{\alpha_0}$ equals $\int_{\mathbb{R}^k_{\ge 0}}\mathbf{D}(\alpha,\mathbf{t})\cdot{\langle}{W}{\rangle}(\mathbf{t})\,\mathrm{d}\mathbf{t}$, where
\[
\mathbf{D}(\alpha,\mathbf{t}):=\prod_{n=0}^{k-1}\left\{\mathbf{P}(\alpha_n,\alpha_{n+1})\cdot\Lambda_{\alpha_n}(\mathbf{t}_n)\right\}\enskip.
\]
\end{lemma}
\begin{proof}
Let $\alpha\in S^{k+1}$ and $W\subseteq\mathbb{R}_{\ge 0}^k$ closed with $k\ge 1$. For every $\epsilon>0$, define the hypercube set $H_k^\epsilon\subseteq\mathcal{P}(\mathbb{R}^k_{\ge 0})$ as follows:
\[
H_k^\epsilon:=\{[m_0\epsilon,(m_0+1)\epsilon]\times\dots\times
[m_{k-1}\epsilon,(m_{k-1}+1)\epsilon]\mid m_n\in\mathbb{N}_0\mbox{
for }0\le n\le k-1\}
\]
When equipped with $\alpha$, each hypercube $\prod_{n=0}^{k-1}[m_n\epsilon,(m_n+1)\epsilon]$ corresponds to the template $\alpha_0[m_0\epsilon,(m_0+1)\epsilon]\dots\alpha_{k-1}[m_{k-1}\epsilon,(m_{k-1}+1)\epsilon]\alpha_k$, which in turn corresponds to a cylinder set. Now define $C_k^\epsilon$ to be a hypercube-cover of $W$ by:
\[
C_k^\epsilon:=\bigcup\{\vartheta\in H_k^\epsilon\mid\vartheta\cap W\ne\emptyset\}\enskip.
\]
Further define $C_k:=\bigcap_{n\in\mathbb{N}}C_{k}^{\epsilon_n}$ where $\epsilon_n=(\frac{1}{2})^n$. We prove that $W=C_k$. It is clear that $W\subseteq C_k$. Suppose that $C_k\not\subseteq W$. Then there is a vector $\xi\in C_k\backslash W$. Since $W$ is a closed set, there exists a neighbourhood around $\xi$ of diameter $d$ in which all vectors are not in $W$. Then $\xi\not\in C_k$ since $\xi\not\in C_{k}^{\epsilon}$ for all $\epsilon<\frac{d}{2\sqrt{k}}$. Contradiction. Thus $W=C_k$. Then it follows from $\bigcap_{n}\mathsf{Path}[\alpha,C_k^{\epsilon_n}]=\mathsf{Path}[\alpha,C_k]$ that $\mathsf{Path}[\alpha,W]$ is measurable under $\Omega$.

We have shown that $W=\bigcap_{n\in\mathbb{N}}C^{\epsilon_n}_k$. Moreover, $\{C^{\epsilon_n}_k\}_{n\in\mathbb{N}}$ is monotonically decreasing since $\epsilon_n=(\frac{1}{2})^n$. Thus ${\langle}{W}{\rangle}(\mathbf{t})=\lim\limits_{n\rightarrow\infty}{\langle}{C}_k^{\epsilon_n}{\rangle}(\mathbf{t})$ for all $\mathbf{t}\in\mathbb{R}^k_{\ge 0}$. Note that
\[
\int_{\mathbb{R}^k_{\ge 0}}\left|\mathbf{D}(\alpha,\mathbf{t})\cdot\begin{pmatrix}
{\langle}{W}{\rangle}(\mathbf{t})\\{\langle}{C}_k^{\epsilon_n}{\rangle}(\mathbf{t})
\end{pmatrix}\right|\,\mathrm{d}\mathbf{t}\le\int_{\mathbb{R}^k_{\ge 0}}\mathbf{D}(\alpha,\mathbf{t})\,\mathrm{d}\mathbf{t}=\prod_{n=0}^{k-1}\,\mathbf{P}(\alpha_n,\alpha_{n+1})<\infty
\]
Then we have $\lim\limits_{n\rightarrow\infty}\int_{\mathbb{R}^{k}_{\ge 0}}\mathbf{D}(\alpha,\mathbf{t})\cdot{\langle}{C}_k^{\epsilon_n}{\rangle}(\mathbf{t})\,\mathrm{d}\mathbf{t}=\int_{\mathbb{R}^{k}_{\ge 0}}\mathbf{D}(\alpha,\mathbf{t})\cdot{\langle}{W}{\rangle}(\mathbf{t})\,\mathrm{d}\mathbf{t}$ by Dominated Convergence Theorem. One can verify that $\int_{\mathbb{R}^{k}_{\ge 0}}\mathbf{D}(\alpha,\mathbf{t})\cdot{\langle}{C}_k^{\epsilon_n}{\rangle}(\mathbf{t})\,\mathrm{d}\mathbf{t}$ equals the probability mass of $\mathsf{Path}\left[\alpha, C_k^{\epsilon_n}\right]$ under $\mathcal{B}_{\alpha_0}$. Thus the probability mass of $\mathsf{Path}[\alpha, W]$ under $\mathcal{B}_{\alpha_0}$ equals $\int_{\mathbb{R}^k_{\ge 0}}\mathbf{D}(\alpha,\mathbf{t})\cdot{\langle}{W}{\rangle}(\mathbf{t})\,\mathrm{d}\mathbf{t}$.
\end{proof}
We handle the measurability result and the system of integral equations simultaneously in the following theorem. Below we define $\mathcal{V}:=S\times Q\times\mathrm{Val}(\mathcal{X})$.

\begin{theorem}\label{thm:mrbl}
For all $(s,q,\eta)\in\mathcal{V}$ and $k\ge 0$, $\mathsf{Path}_k(s,q,\eta)$ is measurable under $\Omega$. Furthermore, the family $\{\mathsf{prob}_k:\mathcal{V}\mapsto [0,1]\}_{k\ge 0}$, where $\mathsf{prob}_k(s,q,\eta)$ is the probability mass of $\mathsf{Path}_k(s,q,\eta)$ under $\mathcal{B}_s$, satisfies the following properties: $\mathsf{prob}_0(s,q,\eta)={\langle}{F}{\rangle}(q)$; If $q\in F$ then $\mathsf{prob}_{k+1}(s,q,\eta)=1$, otherwise
\begin{align*}
&\mathsf{prob}_{k+1}(s,q,\eta)=\\
&\quad\int_0^{+\infty}\left\{\Lambda_s(t)\cdot\left[\sum_{u\in
S}\mathbf{P}(s,u)\cdot\mathsf{prob}_k\left(u,\mathbf{q}_{q,s}^{\eta+t},(\eta+t)[\mathbf{X}_{q,s}^{\eta+t}:=0]\right)\right]\right\}\mathrm{d}t
\end{align*}
\end{theorem}
\begin{proof}
First we prove that $\mathsf{Path}_k(s,q,\eta)$ is measurable. Let $(s,q,\eta)\in\mathcal{V}$. The case $k=0$ is easy: $\mathsf{Path}_0(s,q,\eta)$ is either $\emptyset$ or $\{\pi\in\Omega\mid \pi[0]=s\}$, depending on whether $q\not\in F$ or $q\in F$. We prove the case when $k\ge 1$. By Lemma~\ref{lem:mrbl:dcmp}, it suffices to prove that each $\mathsf{Path}^{\alpha}_{\beta,\eta}$ with $(\alpha,\beta)\in D^{k}(q,s)$ is measurable.

Let $(\alpha,\beta)\in D^{k}(q,s)$. By Lemma~\ref{lem:dcmp:char}, $\mathsf{Path}^{\alpha}_{\beta,\eta}=\mathsf{Path}[\alpha,J_{\beta,\eta}]$. As is mentioned previously, $J_{\beta,\eta}$ is specified by a finite conjunctive collection of linear constraints on $\{\pi{\langle}{n}{\rangle}\}_{0\le n\le k-1}$: each takes the form $\sum_{n=l_1}^{l_2}\pi{\langle}{n}{\rangle}\Join c$ where $0\le l_1\le l_2\le k-1$, $\Join\in\{<,\le,>,\ge\}$ and $c\in\mathbb{R}$. We distinguish two cases below.

\textbf{Case 1}: All $\Join$'s present in the linear constraints are either $\le$ or $\ge$. Then $J_{\beta,\eta}$ is closed  in $\mathbb{R}^k_{\ge 0}$. Thus by Lemma~\ref{lem:mrbl:closed}, $\mathsf{Path}[\alpha,J_{\beta,\eta}]$ is measurable under $\Omega$.

\textbf{Case 2}: Some $\Join$ is $<$ or $>$. The point is that ``${<}{c}=\bigcup\{{\le}{c-\epsilon}\mid
\epsilon>0\}$'' and ``${>}{c}$'' likewise. Thus by the fact that $J_{\beta,\eta}$ is specified by a finite number of linear constraints, we have $J_{\beta,\eta}=\bigcup_{n\in\mathbb{N}}{J}^{n}_{\beta,\eta}$ where ${J}^n_{\beta,\eta}$ is specified by the set of constraints obtained from $J_{\beta,\eta}$ by replacing
each occurrence of ``${<}{c}$'' with ``${\le}{c-{({\frac{1}{2}})}^n}$'' and ``${>}{c}$'' likewise. Because each $J^n_{\beta,\eta}$ is closed, $\mathsf{Path}[\alpha,J^n_{\beta,\eta}]$ is measurable under $\Omega_s$. Then by $\mathsf{Path}[\alpha,J_{\beta,\eta}]=\bigcup_{n\in\mathbb{N}}\mathsf{Path}[\alpha,J^n_{\beta,\eta}]$, we obtain that $\mathsf{Path}^{\alpha}_{\beta,\eta}$ is measurable under $\Omega$.

Now we prove the integral-equation system for $\mathsf{prob}$. Let $(s,q,\eta)\in\mathcal{V}$. By definition, we have $\mathsf{prob}_0(s,q,\eta)={\langle}{F}{\rangle}(q)$ and $\mathsf{prob}_{k+1}(s,q,\eta)=1$ if $q\in F$. We prove the relation between $\mathsf{prob}_{k+1}$ and $\mathsf{prob}_k$ when $q\not\in F$. By Lemma~\ref{lem:mrbl:dcmp}, $\mathsf{Path}_{k+1}(s,q,\eta)$ is the disjoint union of $\{\mathsf{Path}^\alpha_{\beta,\eta}\mid(\alpha,\beta)\in D^{k+1}(q,s)\}$. Then
$\mathsf{prob}_{k+1}(s,q,\eta)=\sum_{(\alpha,\beta)\in D^{k+1}(q,s)}\mathsf{prob}^{\alpha}_{\beta,\eta}$\enskip,
where $\mathsf{prob}^{\alpha}_{\beta,\eta}$ is the probability mass of $\mathsf{Path}^{\alpha}_{\beta,\eta}$ under $\mathcal{B}_{\alpha_0}$. We first prove that:
\begin{equation*}
\mathsf{prob}^{\alpha}_{\beta,\eta}=\int_{\mathbb{R}^{m}_{\ge 0}}\mathbf{D}(\alpha,\mathbf{t})\cdot {\langle}{J}_{\beta,\eta}{\rangle}(\mathbf{t})\,\mathrm{d}\mathbf{t}\tag{$\dag$}
\end{equation*}
given any $m\ge 1$ and $(\alpha,\beta)\in D^{m}$. Analogously, we distinguish two cases based on the types of constraints $\Join$ that specify $J_{\beta,\eta}$.

\textbf{Case 1}: All $\Join$'s are either $\le$ or $\ge$. Then the result follows from Lemma~\ref{lem:mrbl:closed}.

\textbf{Case 2}: Some $\Join$ is $<$ or $>$. We have shown that $J_{\beta,\eta}=\bigcup_{n\in\mathbb{N}}J^n_{\beta,\eta}$, where $J^n_{\beta,\eta}$ is obtained from $J_{\beta,\eta}$ by relaxing $<$ and $>$ with $(\frac{1}{2})^n$. Furthermore, $\lim\limits_{n\rightarrow\infty}{\langle}{J}^n_{\beta,\eta}{\rangle}(\mathbf{t})={\langle}{J}_{\beta,\eta}{\rangle}(\mathbf{t})$ because $\{J^n_{\beta,\eta}\}_{n\ge 1}$ is monotonically increasing. By Lemma~\ref{lem:mrbl:closed}, $\mathcal{B}_{\alpha_0}\left(\mathsf{Path}[\alpha,J^n_{\beta,\eta}]\right)$ equals $\int_{\mathbb{R}^{m}_{\ge 0}}\mathbf{D}(\alpha,\mathbf{t})\cdot{\langle}{J}^n_{\beta,\eta}{\rangle}(\mathbf{t})\,\mathrm{d}\mathbf{t}$. Thus by Dominated Convergence Theorem, we obtain $(\dag)$.

Consider $\mathsf{prob}^{\alpha}_{\beta,\eta}$ where $k\ge 1$, $(\alpha,\beta)\in D^{k+1}$ and $\mathfrak{q}(\beta_0)\not\in F$. Define
\[
\hat{\alpha}:=\alpha_1\dots\alpha_{k+1}\mbox{ and }\hat{\beta}:=\beta_1\dots\beta_{k}
\]
By Lemma~\ref{lem:cons:recr},
${\langle}{J}_{\beta,\eta}{\rangle}(\mathbf{t})={\langle}{\mathfrak{g}(\beta_0)}{\rangle}(\eta+\mathbf{t}_0)\cdot{\langle}{J}_{\hat{\beta},(\eta+\mathbf{t}_0)[\mathfrak{X}(\beta_0):=0]}{\rangle}(\hat{\mathbf{t}})$ for all $\mathbf{t}\in\mathbb{R}^{k+1}_{\ge 0}$, where $\hat{\mathbf{t}}:=(\mathbf{t}_1,\dots,\mathbf{t}_{k})$. Then by Fubini's Theorem and $(\dag)$, we have
\begin{eqnarray*}
\mathsf{prob}^{\alpha}_{\beta,\eta}&=&\int_{\mathbb{R}^{k+1}_{\ge 0}}\mathbf{D}(\alpha,\mathbf{t})\cdot{\langle}{J}_{\beta,\eta}{\rangle}(\mathbf{t})\,\mathrm{d}\mathbf{t}\\
&=&\int_{\mathbb{R}_{\ge 0}}\mathbf{D}(\alpha_0\alpha_1,t)\cdot
{\langle}{\mathfrak{g}(\beta_0)}{\rangle}(\eta+t)\cdot\\
& &\qquad\left(\int_{\mathbb{R}^{k}_{\ge 0}}\mathbf{D}(\hat{\alpha},\hat{\mathbf{t}})\cdot{\big\langle}{J}_{\hat{\beta},(\eta+t)[\mathfrak{X}(\beta_0):=0]}{\big\rangle}(\hat{\mathbf{t}})\,\mathrm{d}\hat{\mathbf{t}}\right)\,\mathrm{d}t\\
&=&\int_{\mathbb{R}_{\ge 0}}\mathbf{D}(\alpha_0\alpha_1,t)\cdot
\langle{\mathfrak{g}(\beta_0)}\rangle(\eta+t)\cdot\mathsf{prob}^{\hat{\alpha}}_{\hat{\beta},(\eta+t)[\mathfrak{X}(\beta_0):=0]}\,\mathrm{d}t\\
\end{eqnarray*}
where in the last step, we use the fact that $(\hat{\alpha},\hat{\beta})\in D^k$. Below we prove the relation between $\mathsf{prob}_{k+1}(s,q,\eta)$ and $\mathsf{prob}_k$ when $q\not\in F$. If $k\ge 1$, we have:
\begin{eqnarray*}
& &\mathsf{prob}_{k+1}(s,q,\eta)\\
&=&\sum_{(\alpha,\beta)\in D^{k+1}(q,s)}\mathsf{prob}^{\alpha}_{\beta,\eta}\\
&=&\sum_{(\alpha,\beta)\in D^{k+1}(q,s)}\int_{\mathbb{R}_{\ge 0}}\mathbf{D}(\alpha_0\alpha_1,t)\cdot
\langle{\mathfrak{g}(\beta_0)}\rangle(\eta+t)\cdot\mathsf{prob}^{\hat{\alpha}}_{\hat{\beta},(\eta+t)[\mathfrak{X}(\beta_0):=0]}\,\mathrm{d}t\\
&=&\sum_{u\in S}\sum_{\gamma\in\Delta_{q,s}}\left\{\sum_{(\alpha,\beta)\in D^{k}(\mathfrak{q}'(\gamma),u)}\left[\int_{\mathbb{R}_{\ge 0}}\mathbf{D}(su,t)\cdot{\langle}{\mathfrak{g}(\gamma)}{\rangle}(\eta+t)\cdot\mathsf{prob}^{\alpha}_{\beta,(\eta+t)[\mathfrak{X}(\gamma):=0]}\,\mathrm{d}t\right]\right\}\\
&=&\sum_{u\in S}\sum_{\gamma\in\Delta_{q,s}}\int_{\mathbb{R}_{\ge 0}}\mathbf{D}(su,t)\cdot{\langle}{\mathfrak{g}(\gamma)}{\rangle}(\eta+t)\cdot\left(\sum_{(\alpha,\beta)\in D^{k}(\mathfrak{q}'(\gamma),u)}\mathsf{prob}^{\alpha}_{\beta,(\eta+t)[\mathfrak{X}(\gamma):=0]}\right)\,\mathrm{d}t\\
&=&\sum_{u\in S}\sum_{\gamma\in\Delta_{q,s}}\int_{\mathbb{R}_{\ge 0}}\mathbf{P}(s,u)\cdot\Lambda_s(t)
\cdot{\langle}{\mathfrak{g}(\gamma)}{\rangle}(\eta+t)\cdot\mathsf{prob}_{k}(u,\mathfrak{q}'(\gamma),(\eta+t)[\mathfrak{X}(\gamma):=0])\,\mathrm{d}t\\
&=&\int_{\mathbb{R}_{\ge 0}}\sum_{\gamma\in\Delta_{q,s}}\left\{\Lambda_s(t)
\cdot{\langle}{\mathfrak{g}(\gamma)}{\rangle}(\eta+t)\cdot\left[\sum_{u\in S}\mathbf{P}(s,u)\cdot\mathsf{prob}_{k}(u,\mathfrak{q}'(\gamma),(\eta+t)[\mathfrak{X}(\gamma):=0])\right]\right\}\,\mathrm{d}t\\
&=&\int_0^{+\infty}\left\{\Lambda_s(t)\cdot\left[\sum_{u\in
S}\mathbf{P}(s,u)\cdot\mathsf{prob}_{k}\left(u,\mathbf{q}_{q,s}^{\eta+t},(\eta+t)[\mathbf{X}_{q,s}^{\eta+t}:=0]\right)\right]\right\}\,\mathrm{d}t
\end{eqnarray*}
where $\Delta_{q,s}:=\{\gamma\in\Delta\mid (\mathfrak{q}(\gamma),\mathfrak{a}(\gamma))=(q,\mathcal{L}(s))\}$ and the last step is obtained by the fact that the integrand functions are identical. If $k=0$, we have:
\begin{eqnarray*}
& &\mathsf{prob}_{1}(s,q,\eta)\\
&=&\sum_{(\alpha,\beta)\in D^{1}(q,s)}\mathsf{prob}^{\alpha}_{\beta,\eta}\\
&=&\sum_{(\alpha,\beta)\in D^{1}(q,s)}\int_{\mathbb{R}_{\ge 0}}\mathbf{D}(\alpha_0\alpha_1,t)\cdot
{\langle}{{J}_{\beta,\eta}}{\rangle}(t)\,\mathrm{d}t\\
&=&\sum_{u\in S}\sum_{\gamma\in\Delta^F_{q,s}}\int_{0}^{+\infty}\mathbf{P}(s,u)\cdot\Lambda_s(t)\cdot{\langle}{\mathfrak{g}(\gamma)}{\rangle}(\eta+t)\,\mathrm{d}t\\
&=&\int_0^{+\infty}\sum_{\gamma\in\Delta^F_{q,s}}\big\{
\Lambda_s(t)\cdot{\langle}{\mathfrak{g}(\gamma)}{\rangle}(\eta+t)\big\}\,\mathrm{d}t\\
&=&\int_0^{+\infty}\left\{\Lambda_s(t)\cdot\left[\sum_{u\in
S}\mathbf{P}(s,u)\cdot\mathsf{prob}_0\left(u,\mathbf{q}_{q,s}^{\eta+t},(\eta+t)[\mathbf{X}_{q,s}^{\eta+t}:=0]\right)\right]\right\}\,\mathrm{d}t
\end{eqnarray*}
where $\Delta_{q,s}^F:=\{\gamma\in\Delta\mid (\mathfrak{q}(\gamma),\mathfrak{a}(\gamma))=(q,\mathcal{L}(s))
\mbox{ and }\mathfrak{q}'(\gamma)\in F\}$ and the last equality is derived from the fact that the integrand functions are identical.
\end{proof}
The main result of this section is as follows.
\begin{corollary}\label{crlly:mrbl}
For all $(s,q,\eta)\in\mathcal{V}$,
$\mathsf{Path}(s,q,\eta)$ is measurable under $\Omega$.
Furthermore, the function $\mathsf{prob}:\mathcal{V}\mapsto [0,1]$,
for which $\mathsf{prob}(s,q,\eta)$ is the probability mass of $\mathsf{Path}(s,q,\eta)$ under
$\mathcal{B}_s$, satisfies the following system of integral equations: If $q\in F$ then $\mathsf{prob}(s,q,\eta)=1$, otherwise
\begin{align*}
&\mathsf{prob}(s,q,\eta)=\\
&\quad\int_0^{+\infty}\left\{\Lambda_s(t)\cdot\left[\sum_{u\in S}\mathbf{P}(s,u)\cdot
\mathsf{prob}\left(u,\mathbf{q}_{q,s}^{\eta+t},(\eta+t)[\mathbf{X}_{q,s}^{\eta+t}:=0]\right)\right]\right\}\,\mathrm{d}t
\end{align*}
\end{corollary}
\begin{proof}
It is clear that $\mathsf{prob}(s,q,\eta)=1$ if $q\in F$. Suppose $q\not\in F$, then by Theorem~\ref{thm:mrbl},
\begin{align*}
&\mathsf{prob}_{k+1}(s,q,\eta)=\\
&\quad\int_0^{+\infty}\left\{\Lambda_s(t)\cdot\left[\sum_{u\in
S}\mathbf{P}(s,u)\cdot\mathsf{prob}_k\left(u,\mathbf{q}_{q,s}^{\eta+t},(\eta+t)[\mathbf{X}_{q,s}^{\eta+t}:=0]\right)\right]\right\}\,\mathrm{d}t
\end{align*}
Note that $\lim_{k\rightarrow\infty}\mathsf{prob}_k=\mathsf{prob}$. Thus by Monotone Convergence Theorem, we obtain the desired result by passing the $\lim$ operator into the integral.
\end{proof}

\section{Equivalences, Lipschitz Continuity and The Product Region Graph}

In this section we prepare several tools to derive the differential characterization for the function $\mathsf{prob}$. In detail, we review several equivalence relations on clock valuations~\cite{DBLP:journals/tcs/AlurD94} and the product region graph between CTMC and DTA~\cite{DBLP:journals/corr/abs-1101-3694}, and derive a Lipschitz Continuity of the function $\mathsf{prob}$.

Below we fix a CTMC $\mathcal{M}=(S,L,\mathbf{P},\lambda,\mathcal{L})$ and a DTA $\mathcal{A}=(Q,L,\mathcal{X},\Delta,F)$. We denote by $T^\mathcal{A}_x$ the largest number $c$ that appears in some guard $x\Join c$ of $\mathcal{A}$ on clock $x$, by $T^\mathcal{A}_{\max}$ the number $\max_{x\in\mathcal{X}}T^\mathcal{A}_x$, and by $\lambda^\mathcal{M}_{\max}$ the value $\max\{\lambda(s)\mid s\in S\}$. We omit $\mathcal{M}$ or $\mathcal{A}$ if the context is clear.

\subsection{Equivalence Relations}

\begin{definition}\cite{DBLP:journals/tcs/AlurD94}
Two valuations $\eta,\eta'\in\mathrm{Val}(\mathcal{X})$ are \emph{guard-equivalent}, denoted by $\eta\equiv_{\mathrm{g}}\eta'$, if they satisfy the following conditions:
\begin{enumerate} \itemsep0pt \parskip0pt \parsep0pt \topsep0pt
\item for all $x\in\mathcal{X}$, $\eta(x)>T_x$ iff $\eta'(x)>T_x$;
\item for all $x\in\mathcal{X}$, if $\eta(x)\le T_x$ and $\eta'(x)\le T_x$, then (i) $\mathrm{int}(\eta(x))=\mathrm{int}(\eta'(x))$ and (ii) $\mathrm{frac}(\eta(x))>0$ iff $\mathrm{frac}(\eta'(x))>0$.
\end{enumerate}
where $\mathrm{int}(),\mathrm{frac}()$ are the integral and fractional part of a real number, respectively. Moreover, $\eta$ and $\eta'$ are \emph{equivalent}, denoted by $\eta\sim\eta'$, if (i) $\eta\equiv_{\mathrm{g}}\eta'$ and (ii) for all $x,y\in\mathcal{X}$, if $\eta(x),\eta'(x)\le T_x$ and $\eta(y),\eta'(y)\le T_y$, then
$\mathrm{frac}(\eta(x))\Join\mathrm{frac}(\eta(y))$ iff $\mathrm{frac}(\eta'(x))\Join\mathrm{frac}(\eta'(y))$ for all $\Join\in\{<,=,>\}$. We will call equivalence classes of $\sim$ \emph{regions}. Given a region $[\eta]_\sim$, we say that $[\eta]_\sim$ is \emph{marginal} if $\eta(x)\le T_x$ and $\mathrm{frac}(\eta(x))=0$ for some clock $x$.
\end{definition}
In other words, equivalence classes of $\equiv_\mathrm{g}$ are captured by a boolean vector over $\mathcal{X}$ which indicates whether $\eta(x)>T_x$, an integer vector which indicates the integral parts on $\eta(x)\le T_x$ and a boolean vector which indicates whether $\eta(x)$ is an integer when $\eta(x)\le T_x$; equivalence classes of $\sim$ is further captured by a linear order on the set $\{x\in\mathcal{X}\mid\eta(x)\le T_x\}$ w.r.t $\mathrm{frac}(\eta(x))$. Below we state some basic properties of $\equiv_{\mathrm{g}}$ and $\sim$.

\begin{proposition}\cite{DBLP:journals/tcs/AlurD94}
The following properties on $\equiv_{\mathrm{g}}$ and $\sim$ hold:
\begin{enumerate} \itemsep0pt \parskip0pt \parsep0pt \topsep0pt
\item Both $\equiv_\mathrm{g}$ and $\sim$ is an equivalence relation over clock valuations, and has finite index;
\item if $\eta\equiv_\mathrm{g}\eta'$ then they satisfy the same set of guards that appear in $\mathcal{A}$;
\item If $\eta\sim\eta'$ then
\begin{itemize} \itemsep0pt \parskip0pt \parsep0pt \topsep0pt
\item for all $t>0$, there exists $t'>0$ such that $\eta+t\sim\eta'+t'$, and
\item for all $t'>0$, there exists $t>0$ such that $\eta+t\sim\eta'+t'$.
\end{itemize}
\item If $\eta\sim\eta'$, then $\eta[X:=0]\sim\eta'[X:=0]$ for all $X\subseteq\mathcal{X}$. Moreover, for all $\eta\in\mathrm{Val}(\mathcal{X})$ and $X\subseteq\mathcal{X}$, $\{\eta'[X:=0]\mid\eta'\in [\eta]_\sim\}$ is a region.
\end{enumerate}
\end{proposition}
Besides these two equivalence notions, we define another finer equivalence notion as follows.

\begin{definition}
Two valuations $\eta,\eta'\in\mathrm{Val}(\mathcal{X})$ are \emph{bound-equivalent}, denoted by $\eta\equiv_{\mathrm{b}}\eta'$, if for all $x\in\mathcal{X}$, either $\eta(x)>T_x$ and $\eta'(x)>T_x$, or $\eta(x)=\eta'(x)$.
\end{definition}
It is straightforward to verify that $\equiv_{\mathrm{b}}$ is an equivalence relation. The following lemma specifies the relation between $\equiv_{\mathrm{b}}$ and $\mathsf{prob}$, see Barbot~\emph{et al.}~\cite{DBLP:conf/tacas/BarbotCHKM11}. Below we present an alternative proof for integrity.

\begin{proposition}\label{lem:boundness}
Let $s\in S$, $q\in Q$ and $\eta,\eta'\in\mathrm{Val}(\mathcal{X})$. If $\eta\equiv_{\mathrm{b}}\eta'$,
then $\mathsf{prob}(s,q,\eta)=\mathsf{prob}(s,q,\eta')$.
\end{proposition}
\begin{proof}
We prove that $\mathsf{Path}(s,q,\eta)=\mathsf{Path}(s,q,\eta')$. Suppose $\pi\in\mathsf{Path}(s,q,\eta)$. Then the run $\mathcal{A}_{q,\eta}(\mathcal{L}_\pi)=\{(q_n,\eta_n)(\mathcal{L}(\pi[n]),\pi\langle n\rangle)\}_{n\ge 0}$ satisfies that $q_n\in F$ for some $n\ge 0$. Denote $\mathcal{A}_{q,\eta'}(\mathcal{L}_\pi)=\{(q'_n,\eta'_n)(\mathcal{L}(\pi[n]),\pi\langle n\rangle)\}_{n\ge 0}$. We prove inductively on $n$ that $q_n=q'_n$ and $\eta_n\equiv_\mathrm{b}\eta'_n$ for all $n\ge 0$. This would imply that $\pi\in\mathsf{Path}(s,q,\eta')$. The inductive proof can be carried out by the fact that $\eta_n\equiv_\mathrm{b}\eta'_n$ implies $\eta_n+{\pi}{\langle}{n}{\rangle}\equiv_\mathrm{b}\eta'_n+{\pi}{\langle}{n}{\rangle}$ and
$(\eta_n+{\pi}{\langle}{n}{\rangle})[X:=0]\equiv_\mathrm{b}(\eta'_n+{\pi}{\langle}{n}{\rangle})[X:=0]$ for all $X\subseteq\mathcal{X}$. Thus $\mathsf{Path}(s,q,\eta)\subseteq\mathsf{Path}(s,q,\eta')$. The other direction can be proved symmetrically.
\end{proof}

In the following, we further introduce a useful proposition.

\begin{proposition}\label{lem:rightleft}
For each $\eta\in\mathrm{Val}(\mathcal{X})$, there exists $t_1>0$ such that $\eta+t\sim\eta+t'$ for all $t,t'\in (0,t_1)$. For each $\eta\in\mathrm{Val}(\mathcal{X})$ such that $\eta(x)>0$ for all $x\in\mathcal{X}$, there exists $t_2>0$ such that $\eta-t\sim\eta-t'$ for all $t,t'\in (0,t_2)$.
\end{proposition}
\begin{proof}
Define $\mathcal{R}':=\{\eta(x)-T_x\mid x\in\mathcal{X}\mbox{ and }\eta(x)>T_x\}$. If $\eta(x)>T_x$ for all clocks $x$, then we can choose $t_1$ to be any positive real number and $t_2=\min\mathcal{R}'$. Below we suppose that there is $x\in\mathcal{X}$ such that $\eta(x)\le T_x$.

Define $\mathcal{R}:=\{\mathrm{frac}(\eta(x))\mid x\in\mathcal{X}\mbox{ and }\eta(x)\le T_x\}$. Let $c_1,c_2$ be the maximum and the minimum value of $\mathcal{R}$, respectively. Note that $0\le c_2\le c_1 <1$. Then we can choose $t_1=1-c_1$. The choice of $t_2$ subjects to the two cases below.
\begin{enumerate}  \itemsep0pt \parskip0pt \parsep0pt \topsep0pt
\item $c_2>0$. Then we can choose $t_2=\min\{\{c_2\}\cup\mathcal{R}'\}$.
\item $c_2=0$. If $\mathcal{R}=\{c_2\}$ then we can choose $t_2=\min\{\{1\}\cup\mathcal{R}'\}$. Otherwise, let $c'>c_2$ be the second minimum value in $\mathcal{R}$. Then we can choose $t_2=\min\{\{c'\}\cup\mathcal{R}'\}$.
\end{enumerate}
It is straightforward to verify that $t_1,t_2$ satisfy the desired property.
\end{proof}
We denote $\eta^+$ to be a representative in $\{\eta+t\mid t\in (0,t_1)\}$, and $\eta^-$ to be a representative in $\{\eta-t\mid t\in (0,t_2)\}$, where $t_1,t_2$ are specified in Proposition~\ref{lem:rightleft}. The choice among the representatives will be irrelevant because they are equivalent under $\sim$. Note that if $[\eta]_\sim$ is not marginal, then $[\eta]_\sim=[\eta^+]_\sim=[\eta^-]_\sim$.

\subsection{The Product Region Graph}

We define a qualitative variation of the product region graph proposed by Chen~\emph{et al.}~\cite{DBLP:journals/corr/abs-1101-3694}, mainly to derive a qualitative property of the function $\mathsf{prob}$. The content of this subsection may be covered by the result by Br\'{a}zdil~\emph{et al.}~\cite{DBLP:conf/concur/BrazdilKKKR10}. Even though, we present it for the sake of integrity.

\begin{definition}
The \emph{product region graph} $G^{\mathcal{M}\otimes\mathcal{A}}=(V^{\mathcal{M}\otimes\mathcal{A}},E^{\mathcal{M}\otimes\mathcal{A}})$ of $\mathcal{M}$ and $\mathcal{A}$ is a directed graph defined as follows: $V=S\times Q\times\left(\mathrm{Val}(\mathcal{X})\slash\sim\right)$, $((s,q,r),(s',q',r'))\in E$ iff (i) $\mathbf{P}(s,s')>0$ and (ii) there exists $\eta\in r,\eta'\in r'$ and $t>0$ such that $[\eta+t]_\sim$ is not a marginal region and $(q',\eta')=\kappa((q,\eta),(\mathcal{L}(s),t))$. A vertex $(s,q,r)\in V$ is \emph{final} if $q\in F$.
\end{definition}
We will omit ${\mathcal{M}\otimes\mathcal{A}}$ in $G^{\mathcal{M}\otimes\mathcal{A}}=(V^{\mathcal{M}\otimes\mathcal{A}},E^{\mathcal{M}\otimes\mathcal{A}})$ if the context is clear. The following lemma states the relationship between $\mathsf{prob}$ and the product region graph. Below we define
\[
\mathcal{R}_\eta:=\{0,1\}\cup\{\mathrm{frac}(\eta(x))\mid x\in\mathcal{X}\mbox{ and }\eta(x)\le T_x\}
\]
for each $\eta\in\mathrm{Val}(\mathcal{X})$. Intuitively, $\mathcal{R}_\eta$ captures the fractional values on $\eta$.

\begin{proposition}\label{lem:prg}
For all $(s,q,\eta)\in\mathcal{V}$, $\mathsf{prob}(s,q,\eta)>0$ iff $(s,q,[\eta]_\sim)$  can reach some final vertex in $G$.
\end{proposition}
\begin{proof}
``$\Rightarrow$'': It is clear that $\mathsf{prob}(s,q,\eta)>0$ iff $\mathsf{prob}_k(s,q,\eta)>0$ for some $k\in\mathbb{N}_{0}$. We prove by induction on $k$ that for all $(s,q,\eta)\in\mathcal{V}$, if $\mathsf{prob}_k(s,q,\eta)>0$ then $(s,q,[\eta]_\sim)$ can reach some final vertex in $G$. The base step $k=0$ is easy. Suppose $\mathrm{prob}_{k+1}(s,q,\eta)>0$ with $q\not\in F$. By Theorem~\ref{thm:mrbl}, we can deduce that
\begin{equation}\label{eq:eq1}
\int_0^{+\infty}\left\{
\Lambda_s(t)\cdot\left[\sum_{u\in
S}\mathbf{P}(s,u)\cdot
\mathsf{prob}_k\left(u,\mathbf{q}_{q,s}^{\eta+t},(\eta+t)[\mathbf{X}_{q,s}^{\eta+t}:=0]\right)\right]\right\}\mathrm{d}t>0
\end{equation}
Consider the regions traversed by $\eta+t$ when $t$ goes from $0$ to $+\infty$. Denote $\mathcal{R}_\eta=\{w_0,\dots,w_m\}$ such that $m\ge 1$ and $w_i>w_{i+1}$ for all $0\le i < m$. Note that $w_0=1$ and $w_m=0$. We divide $[0,+\infty)$ into open integer intervals $(0,1),(1,2),(T_{\max}-1,T_{\max}),(T_{\max},\infty)$. For each $n<T_{\max}$, we further divide the interval $(n,n+1)$ into the following open sub-intervals, excluding a finite number of isolating points:
\[
(n+1-w_{0},n+1-w_{1}),~\dots,~(n+1-w_{m-1}, n+1-w_{m})\enskip.
\]
Then we define the cluster
\[
\mathcal{I}:=\{(n+1-w_{i},n+1-w_{i+1})\mid 0\le n<T_{\max}, 0\le i<m\}\cup\{(T_{\max},+\infty)\}\enskip.
\]
One can verify that for all $I\in\mathcal{I}$ and $t',t\in I$, $\eta+t\sim\eta+t'$. In other words, $[\eta+t]_\sim$ does not change when $t$ is restricted to one of the intervals from $\mathcal{I}$. By~(\ref{eq:eq1}), there exists $I\in\mathcal{I}$ such that
\begin{equation*}
\int_I\left\{\Lambda_s(t)\cdot\left[\sum_{u\in S}\mathbf{P}(s,u)\cdot\mathsf{prob}_k\left(u,\mathbf{q}_{q,s}^{\eta+t},(\eta+t)[\mathbf{X}_{q,s}^{\eta+t}:=0]\right)\right]\right\}\mathrm{d}t>0\enskip.
\end{equation*}
This means that there is $u^*\in S$ and $t^*\in I$ such that
\[
\mathbf{P}(s,u^*)\cdot\mathsf{prob}_k\left(u^*,\mathbf{q}_{q,s}^{\eta+t^*},(\eta+t^*)[\mathbf{X}_{q,s}^{\eta+t^*}:=0]\right)>0\enskip.
\]
Since $I$ is nonempty, $[\eta+t^*]_\sim$ is not a marginal region. Thus there is an edge from $(s,q,[\eta]_\sim)$ to $(u^*,\mathbf{q}_{q,s}^{\eta+t^*},[(\eta+t^*)[\mathbf{X}_{q,s}^{\eta+t^*}:=0]]_\sim)$ in $G$. By the induction hypothesis, the vertex $(u^*,\mathbf{q}_{q,s}^{\eta+t^*},[(\eta+t^*)[\mathbf{X}_{q,s}^{\eta+t^*}:=0]]_\sim)$ can reach some final vertex $G$. Then $(s,q,[\eta]_\sim)$ can reach some final vertex in $G$.

``$\Leftarrow$'': Suppose $(s,q,[\eta]_\sim)$ can reach some final vertex in $G$. Let the path be
\[
(s,q,[\eta]_\sim)=(s_k,q_k,r_k)\rightarrow (s_{k-1},q_{k-1},r_{k-1})\dots\rightarrow (s_0,q_0,r_0)
\]
with $q_0\in F$. We prove inductively on $n\le k$ that $\mathsf{prob}_n(s_n,q_n,\eta')>0$ for all $\eta'\in r_n$. The case $n=0$ is clear. Suppose $\mathsf{prob}_n(s_n,q_n,\eta')>0$ for all $\eta'\in r_n$. Let $\eta''\in r_{n+1}$ be an arbitrary valuation. By $(s_{n+1},q_{n+1},r_{n+1})\rightarrow(s_n,q_n,r_n)$, $\mathbf{P}(s_{n+1},s_n)>0$ and there exists $\eta_{n+1}\in r_{n+1}$, $\eta_n\in r_n$ and $t>0$ such that $[\eta_{n+1}+t]_\sim$ is not marginal and $(q_{n},\eta_{n})=\kappa((q_{n+1},\eta_{n+1}),(\mathcal{L}(s_{n+1}),t))$. By $\eta''\sim\eta_{n+1}$, there exists $t'>0$ such that $\eta''+t'\sim\eta_{n+1}+t$, which implies that
\[
(\eta''+t')[\mathbf{X}_{q_{n+1},s_{n+1}}^{\eta''+t'}:=0]\sim(\eta_{n+1}+t)[\mathbf{X}_{q_{n+1},s_{n+1}}^{\eta_{n+1}+t}:=0]~~(=\eta_n)\enskip.
\]
By the fact that $[\eta''+t']_\sim$ is not marginal, there exists an interval $I\subseteq\mathbb{R}_{\ge 0}$ with positive length such that for all $\tau\in I$, $\eta''+\tau\sim\eta''+t'$ and
\[
(\eta''+\tau)[\mathbf{X}_{q_{n+1},s_{n+1}}^{\eta''+\tau}:=0]\sim(\eta''+t')[\mathbf{X}_{q_{n+1},s_{n+1}}^{\eta''+t'}:=0]\sim\eta_n\enskip.
\]
Thus by induction hypothesis, $\mathsf{prob}_n\left(s_n,q_n,(\eta''+\tau)[\mathbf{X}_{q_{n+1},s_{n+1}}^{\eta''+\tau}:=0]\right)>0$ for all $\tau\in I$. Hence
\[
\int_I\left\{\Lambda_s(\tau)\cdot\left[\mathbf{P}(s_{n+1},s_n)\cdot\mathsf{prob}_n\left(s_n,q_n,(\eta''+\tau)[\mathbf{X}_{q_{n+1},s_{n+1}}^{\eta''+\tau}:=0]\right)\right]\right\}\mathrm{d}\tau>0\enskip.
\]
It follows that $\mathrm{prob}_{n+1}(s_{n+1},q_{n+1},\eta'')>0$.
\end{proof}

\subsection{Lipschitz Continuity}

Below we prove a Lipschitz Continuity property of $\mathsf{prob}$. More specifically, we prove that all functions that satisfies a boundness condition related to $\equiv_{\mathrm{b}}$ and the system of integral equations specified in Corollary~\ref{crlly:mrbl} are Lipschitz continuous. The Lipschitz continuity will be fundamental to our differential characterization and the error bound of our approximation result.

\begin{theorem}\label{thm:lipschitz}
Let $h:\mathcal{V}\mapsto [0,1]$ be a function which satisfies the following conditions for all $s\in S$, $q\in Q$ and $\eta,\eta'\in\mathrm{Val}(\mathcal{X})$:
\begin{itemize} \itemsep0pt \parskip0pt \parsep0pt \topsep0pt
\item if $\eta\equiv_{\mathrm{b}}\eta'$ then $h(s,q,\eta)=h(s,q,\eta')$;
\item if $q\in F$ then $h(s,q,\eta)=1$, otherwise $h(s,q,\eta)$ is equal to the integral
\[
\int_0^{+\infty}\left\{\Lambda_s(t)\cdot\left[\sum_{u\in
S}\mathbf{P}(s,u)\cdot h\left(u,\mathbf{q}_{q,s}^{\eta+t},(\eta+t)[\mathbf{X}_{q,s}^{\eta+t}:=0]\right)\right]\right\}\,\mathrm{d}t\enskip.
\]
\end{itemize}
Then for all $s\in S$, $q\in Q$ and $\eta,\eta'\in\mathrm{Val}(\mathcal{X})$, if ${{\parallel}{\eta-\eta'}{\parallel}}_\infty<1$ then
\[
|h(s,q,\eta)-h(s,q,\eta')|\le M_1\cdot{{\parallel}{\eta-\eta'}{\parallel}}_\infty\enskip,
\]
where $M_1:=|\mathcal{X}|\cdot \lambda_{\max}T_{\max}\cdot e^{\lambda_{\max}T_{\max}}$\enskip.
\end{theorem}
\begin{proof}
If $q\in F$, then the result follows from $h(s,q,\eta)=h(s,q,\eta')=1$. From now on we suppose that $q\not\in F$. To prove the theorem, it suffices to prove that
\[
|h(s,q,\eta)-h(s,q,\eta')|\le \lambda_{\max}T_{\max}\cdot e^{\lambda_{\max}T_{\max}}\cdot{{\parallel}{\eta-\eta'}{\parallel}}_\infty
\]
when ${{\parallel}{\eta-\eta'}{\parallel}}_\infty<1$ and $\eta,\eta'$ differ only on one clock, i.e. $|\{x\in\mathcal{X}\mid\eta(x)\ne\eta'(x)\}|=1$. To this end we define $\delta(\epsilon)$ for each $\epsilon\in (0,1)$ as follows:
\begin{eqnarray*}
\delta(\epsilon):=&\sup&\{|h(s,q,\eta)-h(s,q,\eta')|\mid (s,q)\in S\times Q, \eta,\eta'\in\mathrm{Val}(\mathcal{X}),\\
& &{{\parallel}{\eta-\eta'}{\parallel}}_\infty\le\epsilon\mbox{ and }\eta,\eta'\mbox{ differ only on one clock}\}
\end{eqnarray*}
Note that for all $\eta,\eta'\in\mathrm{Val}(\mathcal{X})$ and $X\subseteq\mathcal{X}$:
\begin{itemize} \itemsep0pt \parskip0pt \parsep0pt \topsep0pt
\item if $\eta$ and $\eta'$ differ at most on one clock, then so are $\eta[X:=0]$ and $\eta'[X:=0]$;
\item ${{\parallel}{\eta[X:=0]-\eta'[X:=0]}{\parallel}}_\infty\le{{\parallel}{\eta-\eta'}{\parallel}}_\infty$.
\end{itemize}
Suppose $(s,q)\in S\times Q$ and $\eta,\eta'\in\mathrm{Val}(\mathcal{X})$ which satisfies ${{\parallel}{\eta-\eta'}{\parallel}}_\infty\le\epsilon<1$ and differ only on the clock $x$, i.e., $\eta(x)\ne\eta'(x)$ and $\eta(y)=\eta'(y)$ for all $y\ne x$. W.l.o.g we can assume that $\eta(x)<\eta'(x)$. We clarify two cases below.

\textbf{Case 1}: $\mathrm{int}(\eta(x))=\mathrm{int}(\eta'(x))$. Then by $\eta(x)<\eta'(x)$, we have
$\mathrm{frac}(\eta(x))<\mathrm{frac}(\eta'(x))$. Consider the ``behaviours'' of $\eta+t$ and $\eta'+t$ when $t$ goes from $0$ to $\infty$. We divide $[0,\infty)$ into open integer intervals $(0,1),(1,2),\dots,(T_{\max}-1,T_{\max})$ and $(T_{\max},\infty)$. For each $n<T_{\max}$, we further divide the interval $(n,n+1)$ into the following open sub-intervals:
\[
(n,n+1-\mathrm{frac}(\eta'(x))), (n+1-\mathrm{frac}(\eta'(x)),n+1-\mathrm{frac}(\eta(x))), (n+1-\mathrm{frac}(\eta(x)),n+1)
\]
One can observe that for $t\in(n,n+1-\mathrm{frac}(\eta'(x)))\cup(n+1-\mathrm{frac}(\eta(x)),n+1)$, we have $\eta+t\equiv_\mathrm{g}\eta'+t$, which implies that $\eta+t$ and $\eta'+t$ satisfies the same set of guards in $\mathcal{A}$.
However for $t\in(n+1-\mathrm{frac}(\eta'(x)),n+1-\mathrm{frac}(\eta(x)))$, it may be the case that $\eta+t{\not\equiv}_\mathrm{g}\eta'+t$ due to their difference on clock $x$. Thus the total length for $t$ within $(n,n+1)$ such that $\eta+t {\not\equiv}_\mathrm{g} \eta'+t$ is smaller than $|\eta(x)-\eta'(x)|$.
Thus we have $(\dag)$:
\begin{eqnarray*}
\delta_n &:=&\Bigg|\,\int_n^{n+1}\Lambda_s(t)\cdot\Bigg\{\sum_{u\in S}\bigg[\mathbf{P}(s,u)\cdot
h\Big(u,\mathbf{q}_{q,s}^{\eta+t},(\eta+t)[\mathbf{X}_{q,s}^{\eta+t}:=0]\Big)\\
&&\quad{}-\mathbf{P}(s,u)\cdot h\Big(u,\mathbf{q}_{q,s}^{\eta'+t},(\eta'+t)[\mathbf{X}_{q,s}^{\eta'+t}:=0]\Big)\bigg]\Bigg\}\,\mathrm{d}t\,\Bigg|\\
&\le &\int_n^{n+1}\left\{\lambda(s)e^{-\lambda(s)t}\cdot\delta(\epsilon)\right\}\,\mathrm{d}t+\lambda(s)e^{-\lambda(s)n}\cdot|\eta(x)-\eta'(x)|\\
&\le &\delta(\epsilon)\cdot\int_n^{n+1}\left\{\lambda(s)e^{-\lambda(s)t}\right\}\,\mathrm{d}t+\lambda(s)e^{-\lambda(s)n}\cdot\epsilon
\end{eqnarray*}
Note that for all $t\in (T_{\max},\infty)$ and $X\subseteq\mathcal{X}$, $(\eta+t)[X:=0]\equiv_{\mathrm{b}}(\eta'+t)[X:=0]$. This implies $h(u,\mathbf{q}_{q,s}^{\eta+t},(\eta+t)[\mathbf{X}_{q,s}^{\eta+t}:=0])=h(u,\mathbf{q}_{q,s}^{\eta'+t},(\eta'+t)[\mathbf{X}_{q,s}^{\eta'+t}:=0])$. Therefore we have $(\ddag)$:
\begin{eqnarray*}
&    &|h(s,q,\eta)-h(s,q,\eta')|\\
&\le &\sum_{n=0}^{T_{\max}-1}\delta_n\\
&\le &\delta(\epsilon)\cdot\int_0^{T_{\max}}\left\{\lambda(s)e^{-\lambda(s)t}\right\}\,\mathrm{d}t+\lambda(s)\cdot\epsilon\cdot\sum_{n=0}^{T_{\max}-1}e^{-\lambda(s)n}\\
&\le &\delta(\epsilon)\cdot(1-e^{-\lambda(s)T_{\max}})+\epsilon\cdot\lambda(s)\cdot T_{\max}\\
&\le &\delta(\epsilon)\cdot(1-e^{-\lambda_{\max}T_{\max}})+\epsilon\cdot\lambda_{\max}\cdot T_{\max}
\end{eqnarray*}
\textbf{Case 2}: $\mathrm{int}(\eta(x))<\mathrm{int}(\eta'(x))$. By $|\eta(x)-\eta'(x)|<1$, we have
$\mathrm{int}(\eta(x))+1=\mathrm{int}(\eta'(x))$ and $\mathrm{frac}(\eta'(x))<\mathrm{frac}(\eta(x))$.
Similarly we divide the interval $[0,\infty)$ into integer intervals $(0,1),(1,2),\dots,(T_{\max}-1,T_{\max}),(T_{\max},\infty)$. And in each interval $(n,n+1)$, we divide the interval into the following open sub-intervals:
\[
(n, n+1-\mathrm{frac}(\eta(x))), (n+1-\mathrm{frac}(\eta(x)), n+1-\mathrm{frac}(\eta'(x))), (n+1-\mathrm{frac}(\eta'(x)), n+1)
\]
If $t\in(n+1-\mathrm{frac}(\eta(x)), n+1-\mathrm{frac}(\eta'(x)))$, then $\eta+t\equiv_\mathrm{g}\eta'+t$. And if $t$ lies in either $(n, n+1-\mathrm{frac}(\eta(x)))$ or $(n+1-\mathrm{frac}(\eta'(x)), n+1)$, then it may be the case that $\eta+t{\not\equiv}_\mathrm{g}\eta'+t$. Thus the total length within $(n,n+1)$ such that $\eta+t{\not\equiv}_\mathrm{g}\eta'+t$ is still smaller than $|\eta(x)-\eta'(x)|$. Therefore we can apply the analysis $(\dag)$ and $(\ddag)$, and obtain that
\[
|h(s,q,\eta)-h(s,q,\eta')|\le\delta(\epsilon)\cdot(1-e^{-\lambda_{\max}T_{\max}})+\epsilon\cdot\lambda_{\max}\cdot T_{\max}
\]
Thus by the definition of $\delta(\epsilon)$, we obtain
\[
\delta(\epsilon)\le\delta(\epsilon)\cdot(1-e^{-\lambda_{\max}T_{\max}})+\epsilon\cdot\lambda_{\max}\cdot T_{\max}
\]
which implies $\delta(\epsilon)\le \epsilon\cdot e^{\lambda_{\max}T_{\max}}\cdot\lambda_{\max}\cdot T_{\max}\enskip$. By letting $\epsilon={{\parallel}{\eta-\eta'}{\parallel}}_\infty$, we obtain the desired result.
\end{proof}

\begin{corollary}\label{crlly:lipschitz}
For all $(s,q)\in S\times Q$ and $\eta,\eta'\in\mathrm{Val}(\mathcal{X})$, if ${{\parallel}{\eta-\eta'}{\parallel}}_\infty<1$ then
\[
|\mathsf{prob}(s,q,\eta)-\mathsf{prob}(s,q,\eta')|\le M_1\cdot{{\parallel}{\eta-\eta'}{\parallel}}_\infty
\]
where $M_1$ is defined as in Theorem~\ref{thm:lipschitz}.
\end{corollary}
\begin{proof}
Directly from Corollary~\ref{crlly:mrbl}, Proposition~\ref{lem:boundness} and Theorem~\ref{thm:lipschitz}.
\end{proof}
By Lipschitz Continuity, we can further prove that $\mathsf{prob}$ is the unique solution of a revised system of integral equations from the one specified in Corollary~\ref{crlly:mrbl}.

\begin{theorem}\label{thm:unique:int}
The function $\mathsf{prob}$ is the unique solution of the following system of integral equations on $h:\mathcal{V}\mapsto[0,1]$:
\begin{enumerate} \itemsep0pt \parskip0pt \parsep0pt \topsep0pt
\item for all $s\in S$, $q\in Q$ and $\eta,\eta'\in\mathrm{Val}(\mathcal{X})$, if $\eta\equiv_{\mathrm{b}}\eta'$ then $h(s,q,\eta)=h(s,q,\eta')$;
\item for all $ (s,q,\eta)\in\mathcal{V}$, $h(s,q,\eta)=0$ if $(s,q,[\eta]_\sim)$ cannot reach a final vertex in $G$, and $h(s,q,\eta)=1$ if $q\in F$;
\item for all $ (s,q,\eta)\in\mathcal{V}$, if $(s,q,[\eta]_\sim)$ can reach a final vertex in $G$ and $q\not\in F$, then $h(s,q,\eta)$ equals
\[
\int_0^{+\infty}\left\{\Lambda_s(t)\cdot\left[\sum_{u\in S}\mathbf{P}(s,u)\cdot
h\left(u,\mathbf{q}_{q,s}^{\eta+t},(\eta+t)[\mathbf{X}_{q,s}^{\eta+t}:=0]\right)\right]\right\}\,\mathrm{d}t~.
\]
\end{enumerate}
\end{theorem}
\begin{proof}
By Corollary~\ref{crlly:mrbl}, Proposition~\ref{lem:prg} and Proposition~\ref{lem:boundness}, $\mathsf{prob}$ satisfies the referred integral-equation system. Below we prove that the integral-equation system has only one solution.

We first prove that if $h:\mathcal{V}\mapsto [0,1]$ satisfies the integral-equation system, then $h$ satisfies the prerequisite of Theorem~\ref{thm:lipschitz}. We only need to consider $h(s,q,\eta)$ such that $(s,q,[\eta]_\sim)$ cannot reach a final vertex in $G$. Note that $h(s,q,\eta)=0$. From the proof of Proposition~\ref{lem:prg}, we can construct a disjoint open interval cluster $\mathcal{I}$ such that: (i) $\bigcup\mathcal{I}\subseteq\mathbb{R}_{\ge 0}$ and $\mathbb{R}_{\ge 0}\backslash\bigcup\mathcal{I}$ is finite; and (ii) for all $I\in\mathcal{I}$ and $t,t'\in I$, $\eta+t\sim\eta+t'$.
Choose any $t\in\bigcup\mathcal{I}$ and $u\in S$ such that $\mathbf{P}(s,u)>0$. Then $(u,\mathbf{q}_{q,s}^{\eta+t},\left[(\eta+t)[\mathbf{X}_{q,s}^{\eta+t}:=0]\right]_\sim)$ cannot reach some final vertex in $G$ since $[\eta+t]_\sim$ is not marginal. Thus $h(u,\mathbf{q}_{q,s}^{\eta+t},(\eta+t)[\mathbf{X}_{q,s}^{\eta+t}:=0])=0$. It follows that
\[
\int_0^{+\infty}\left\{\Lambda_s(t)\cdot\left[\sum_{u\in S}\mathbf{P}(s,u)\cdot
h\left(u,\mathbf{q}_{q,s}^{\eta+t},(\eta+t)[\mathbf{X}_{q,s}^{\eta+t}:=0]\right)\right]\right\}\,\mathrm{d}t=0\enskip.
\]

Now suppose $h_1,h_2:\mathcal{V}\mapsto[0,1]$ are two solutions of the integral-equation system. Define $h:=|h_1-h_2|$. Then by Theorem~\ref{thm:lipschitz}, $h$ is continuous on $\mathrm{Val}(\mathcal{X})$. Further by the fact that $h(s,q,\eta)=h(s,q,\eta')$ whenever $\eta\equiv_\mathrm{b}\eta'$, the image of $h$ can be obtained on $S\times Q\times\prod_{x\in\mathcal{X}}[0,T_x]$. Thus the maximum value
\[
M:=\sup\{h(s,q,\eta)\mid (s,q,\eta)\in\mathcal{V}\}
\]
can be reached. Below we prove by contradiction that $M=0$. Suppose $M>0$. Denote $\mathcal{Q}:=\{(s,q,\eta)\in\mathcal{V}\mid h(s,q,\eta)=M\}$\enskip. We first prove that $(\dag)$: for all $(s,q,\eta)\in\mathcal{Q}$ and all edge $(s,q,[\eta]_\sim)\rightarrow (s',q',r')$ in $G$, there exists $\eta'\in r'$ such that $(s',q',\eta')\in\mathcal{Q}$.

Consider an arbitrary $(s,q,\eta)\in\mathcal{Q}$. By $M>0$, we have $(s,q,[\eta]_\sim)$ can reach a final vertex in $G$ and $q\not\in F$. As before, we can divide $[0,+\infty)$ into a cluster $\mathcal{I}$ of open intervals, disregarding only a finite number of isolating points $t$, such that $[\eta+\tau]_\sim$ ($\tau\in I$) does not change for each $I\in\mathcal{I}$. Thus $h\left(u,\mathbf{q}_{q,s}^{\eta+t},(\eta+t)[\mathbf{X}_{q,s}^{\eta+t}:=0]\right)$ is piecewise continuous on $t\in\mathbb{R}_{\ge 0}$, for all $u\in S$. Note that
\begin{eqnarray*}
& &h(s,q,\eta)\\
&\le &\int_0^{+\infty}\left\{\Lambda_s(t)\cdot\left[\sum_{u\in S}\mathbf{P}(s,u)\cdot
h\left(u,\mathbf{q}_{q,s}^{\eta+t},(\eta+t)[\mathbf{X}_{q,s}^{\eta+t}:=0]\right)\right]\right\}\,\mathrm{d}t\\
&\le & \int_0^{+\infty}\left\{\Lambda_s(t)\cdot\left[\sum_{u\in S}\mathbf{P}(s,u)\cdot
h\left(s,q,\eta\right)\right]\right\}\,\mathrm{d}t\\
&=& h\left(s,q,\eta\right)
\end{eqnarray*}
By the piecewise continuity, $h\left(u,\mathbf{q}_{q,s}^{\eta+t},(\eta+t)[\mathbf{X}_{q,s}^{\eta+t}:=0]\right)=M$ whenever $t\in\bigcup\mathcal{I}$ and $\mathbf{P}(s,u)>0$. Note that $[0,+\infty)\backslash\bigcup\mathcal{I}$ is finite. Thus for all edge $(s,q,[\eta]_\sim)\rightarrow (s',q',r')$ in $G$, there exists $t\in\bigcup\mathcal{I}$ such that $\mathbf{q}_{q,s}^{\eta+t}=q'$ and $(\eta+t)[\mathbf{X}_{q,s}^{\eta+t}:=0]\in r'$. It follows from $(s',q', (\eta+t)[\mathbf{X}_{q,s}^{\eta+t}:=0]))\in\mathcal{Q}$ that $(\dag)$ holds.

Let $(s,q,\eta)\in\mathcal{Q}$. Then there exists a path
\[
(s,q,[\eta]_\sim)=(s_0,q_0,r_0)\rightarrow (s_1,q_1,r_1)\dots (s_n,q_n,r_n)
\]
in $G$ with $q_n\in F$. However by $(\dag)$, one can prove through induction that there exists $\eta'\in r_n$ such that $(s_n,q_n,\eta')\in\mathcal{Q}$, which implies $q_n\not\in F$. Contradiction. Thus $M=0$ and $h(s,q,\eta)\equiv 0$.
\end{proof}

\section{A Differential Characterization}

In this section we present a differential characterization for the function $\mathsf{prob}$. We fix a CTMC
$\mathcal{M}=(S,L,\mathbf{P},\lambda,\mathcal{L})$ and a DTA $\mathcal{A}=(Q,L,\mathcal{X},\Delta,F)$. Below we introduce our notion of derivative, which is a directional derivative as follows.

\begin{definition}
Given a function $h:\mathcal{V}\mapsto [0,1]$, we denote by $\nabla^+_\mathbf{1}h$ and resp. $\nabla^-_\mathbf{1}h$ the right directional derivative and resp. the left directional derivative of $h$ along the direction $\mathbf{1}$ if the derivative exists. Formally, we define
\begin{itemize} \itemsep0pt \parskip0pt \parsep0pt \topsep0pt
\item $\nabla^+_\mathbf{1}h(s,q,\eta):=\lim\limits_{t\rightarrow 0^+}\frac{1}{t}\cdot (h(s,q,\eta+t)-h(s,q,\eta))$, if the limit exists.
\item $\nabla^-_\mathbf{1}h(s,q,\eta):=\lim\limits_{t\rightarrow 0^+}\frac{1}{t}\cdot (h(s,q,\eta)-h(s,q,\eta-t))$, if  $\eta(x)>0$ for all $x\in\mathcal{X}$ and the limit exists.
\end{itemize}
for each $(s,q,\eta)\in\mathcal{V}$.
\end{definition}
Below we calculate these directional derivatives.

\begin{theorem}\label{thm:drv}
For all $(s,q,\eta)\in\mathcal{V}$ with $q\not\in F$, $\nabla^+_\mathbf{1}\mathsf{prob}(s,q,\eta)$ exists, and $\nabla^-_\mathbf{1}\mathsf{prob}(s,q,\eta)$ exists given that $\eta(x)>0$ for all $x\in\mathcal{X}$. Furthermore, we have
\[
\nabla^+_\mathbf{1}\mathsf{prob}(s,q,\eta)=\lambda(s)\cdot\mathsf{prob}(s,q,\eta)-\lambda(s)\cdot\sum_{u\in S}\mathbf{P}(s,u)\cdot\mathsf{prob}(u,\mathbf{q}^{\eta^+}_{q,s},\eta[\mathbf{X}^{\eta^+}_{q,s}:=0])
\]
and
\[
\nabla^-_\mathbf{1}\mathsf{prob}(s,q,\eta)=\lambda(s)\cdot\mathsf{prob}(s,q,\eta)-\lambda(s)\cdot\sum_{u\in S}\mathbf{P}(s,u)\cdot\mathsf{prob}(u,\mathbf{q}^{\eta^-}_{q,s},\eta[\mathbf{X}^{\eta^-}_{q,s}:=0])
\]
whenever $\nabla^-_\mathbf{1}\mathsf{prob}(s,q,\eta)$ exists.
\end{theorem}
\begin{proof}
We first prove the case for $\nabla^+_\mathbf{1}\mathsf{prob}(s,q,\eta)$. By Corollary~\ref{crlly:mrbl},
\begin{eqnarray*}
\mathsf{prob}(s,q,\eta+t)&=&\int_0^{+\infty}\bigg\{\Lambda_s(\tau)\cdot\Big[\sum_{u\in
S}\mathbf{P}(s,u)\cdot\\
& &\quad\quad\mathsf{prob}\big(u,\mathbf{q}_{q,s}^{\eta+(t+\tau)},(\eta+(t+\tau))[\mathbf{X}_{q,s}^{\eta+(t+\tau)}:=0]\big)\Big]\bigg\}\,\mathrm{d}\tau
\end{eqnarray*}
for $t\ge 0$. Note that the integrand function is Riemann integratable since it is piecewise continuous on $\tau$. By the variable substitution $\tau'=t+\tau$, we have for $t\ge 0$,
\begin{eqnarray*}
\mathsf{prob}(s,q,\eta+t)&=& e^{\lambda(s)\cdot t}\cdot\int_t^{+\infty}\bigg\{
\Lambda_s(\tau)\cdot\Big[\sum_{u\in S}\mathbf{P}(s,u)\cdot \\
& &\quad\quad\mathsf{prob}\big(u,\mathbf{q}_{q,s}^{\eta+\tau},(\eta+\tau)[\mathbf{X}_{q,s}^{\eta+\tau}:=0]\big)\Big]\bigg\}\,\mathrm{d}\tau
\end{eqnarray*}
Then we have
\begin{eqnarray*}
& & \mathsf{prob}(s,q,\eta+t)-\mathsf{prob}(s,q,\eta)\\
&=& e^{\lambda(s)\cdot t}\cdot\int_t^{+\infty}\left\{\Lambda_s(\tau)\cdot\left[\sum_{u\in
S}\mathbf{P}(s,u)\cdot\mathsf{prob}\left(u,\mathbf{q}_{q,s}^{\eta+\tau},(\eta+\tau)[\mathbf{X}_{q,s}^{\eta+\tau}:=0]\right)\right]\right\}\,\mathrm{d}\tau\\
& & {}-\int_0^{+\infty}\left\{
\Lambda_s(\tau)\cdot\left[\sum_{u\in
S}\mathbf{P}(s,u)\cdot
\mathsf{prob}\left(u,\mathbf{q}_{q,s}^{\eta+\tau},(\eta+\tau)[\mathbf{X}_{q,s}^{\eta+\tau}:=0]\right)\right]\right\}\,\mathrm{d}\tau\\
&=& (e^{\lambda(s)\cdot t}-1)\cdot\int_t^{+\infty}\left\{
\Lambda_s(\tau)\cdot\left[\sum_{u\in
S}\mathbf{P}(s,u)\cdot
\mathsf{prob}\left(u,\mathbf{q}_{q,s}^{\eta+\tau},(\eta+\tau)[\mathbf{X}_{q,s}^{\eta+\tau}:=0]\right)\right]\right\}\,\mathrm{d}\tau\\
& &{}-\int_0^t\left\{\Lambda_s(\tau)
\cdot\left[\sum_{u\in
S}\mathbf{P}(s,u)\cdot
\mathsf{prob}\left(u,\mathbf{q}_{q,s}^{\eta+\tau},(\eta+\tau)[\mathbf{X}_{q,s}^{\eta+\tau}:=0]\right)\right]\right\}\,\mathrm{d}\tau
\end{eqnarray*}
By Proposition~\ref{lem:rightleft}, there exists $t_1>0$ such that $\mathbf{q}_{q,s}^{\eta+\tau}$ and $\mathbf{X}_{q,s}^{\eta+\tau}$ does not change for $\tau\in(0,t_1)$. Thus the integrand function in the integral
\[
\int_0^t\left\{
\Lambda_s(\tau)\cdot\left[\sum_{u\in
S}\mathbf{P}(s,u)\cdot
\mathsf{prob}\left(u,\mathbf{q}_{q,s}^{\eta+\tau},(\eta+\tau)[\mathbf{X}_{q,s}^{\eta+\tau}:=0]\right)\right]\right\}\,\mathrm{d}\tau
\]
is continuous on $\tau$ when $t\in (0,t_1)$, and the point $\tau=0$ can be continuously redefined. Thus by L'H\^{o}spital's Rule, we obtain
\[
\nabla^+_\mathbf{1}\mathsf{prob}(s,q,\eta)=\lambda(s)\cdot\mathsf{prob}(s,q,\eta)-\lambda(s)\cdot\sum_{u\in S}\mathbf{P}(s,u)\cdot \mathsf{prob}(u,\mathbf{q}^{\eta^+}_{q,s},\eta[\mathbf{X}^{\eta^+}_{q,s}:=0])
\]
Then we handle the case for $\nabla^-_\mathbf{1}\mathsf{prob}(s,q,\eta)$ given that $\eta(x)>0$ for all $x\in\mathcal{X}$. For $t\in (0,\min\{\eta(x)\mid x\in\mathcal{X}\})$, we have
\begin{eqnarray*}
& & \mathsf{prob}(s,q,\eta)-\mathsf{prob}(s,q,\eta-t)\\
&=& \mathsf{prob}(s,q,(\eta-t)+t)-\mathsf{prob}(s,q,\eta-t)\\
&=& (e^{\lambda(s)\cdot t}-1)\cdot \int_t^{+\infty}\left\{
\Lambda_s(\tau)\cdot\left[\sum_{u\in
S}\mathbf{P}(s,u)\cdot\mathsf{prob}\left(u,\mathbf{q}_{q,s}^{\eta+(\tau-t)},(\eta+(\tau-t))[\mathbf{X}_{q,s}^{\eta+(\tau-t)}:=0]\right)\right]\right\}\,\mathrm{d}\tau\\
 & &{}-\int_0^t\left\{\Lambda_s(\tau)\cdot\left[\sum_{u\in
S}\mathbf{P}(s,u)\cdot\mathsf{prob}\left(u,\mathbf{q}_{q,s}^{\eta-(t-\tau)},(\eta-(t-\tau))[\mathbf{X}_{q,s}^{\eta-(t-\tau)}:=0]\right)\right]\right\}\,\mathrm{d}\tau\\
&=& (1-e^{-\lambda(s)\cdot t})\cdot \int_0^{+\infty}\left\{
\Lambda_s(\tau)\cdot\left[\sum_{u\in
S}\mathbf{P}(s,u)\cdot\mathsf{prob}\left(u,\mathbf{q}_{q,s}^{\eta+\tau},(\eta+\tau)[\mathbf{X}_{q,s}^{\eta+\tau}:=0]\right)\right]\right\}\,\mathrm{d}\tau\\
& &{}-e^{-\lambda(s)\cdot t}\cdot\int_0^t\left\{\lambda(s)e^{\lambda(s)\tau}\cdot\left[\sum_{u\in
S}\mathbf{P}(s,u)\cdot\mathsf{prob}\left(u,\mathbf{q}_{q,s}^{\eta-\tau},(\eta-\tau)[\mathbf{X}_{q,s}^{\eta-\tau}:=0]\right)\right]\right\}\,\mathrm{d}\tau\\
\end{eqnarray*}
where the last step is obtained by performing $\tau'=\tau-t$ in the first integral and $\tau'=t-\tau$ in the second integral.
By Proposition~\ref{lem:rightleft}, there exists $t_2>0$ such that $\mathbf{q}_{q,s}^{\eta-t}$ and $\mathbf{X}_{q,s}^{\eta-t}$ does not change for $t\in (0,t_2)$. Thus the integrand function in the integral
\[
\int_0^t\left\{\lambda(s)e^{\lambda(s)\tau}\cdot\left[\sum_{u\in
S}\mathbf{P}(s,u)\cdot\mathsf{prob}\left(u,\mathbf{q}_{q,s}^{\eta-\tau},(\eta-\tau)[\mathbf{X}_{q,s}^{\eta-\tau}:=0]\right)\right]\right\}\,\mathrm{d}\tau
\]
is continuous on $\tau$ when $t\in (0,t_2)$. Furthermore, the point $\tau=0$ can be continuously redefined. Thus we can also apply L'H\^{o}spital's Rule and obtain the desired result.
\end{proof}
\begin{remark}
Note that if $[\eta]_\sim$ is not marginal, then $[\eta^+]_\sim=[\eta^-]_\sim=[\eta]_\sim$. This tells us that $\nabla_\mathbf{1}\mathsf{prob}(s,q,\eta)$ exists when $[\eta]_\sim$ is not marginal, i.e., $\nabla^+_\mathbf{1}\mathsf{prob}(s,q,\eta)=\nabla^-_\mathbf{1}\mathsf{prob}(s,q,\eta)$.\qed
\end{remark}
Based on Theorem~\ref{thm:drv}, we present our differential characterization.

\begin{theorem}\label{thm:unique:diff}
The function $\mathsf{prob}$ is the unique solution of the following system of differential equations on $h:\mathcal{V}\mapsto[0,1]$: given any $s\in S,q\in Q$ and $\eta,\eta'\in\mathrm{Val}(\mathcal{X})$,
\begin{enumerate} \itemsep0pt \parskip0pt \parsep0pt \topsep0pt
\item if $\eta\equiv_{\mathrm{b}}\eta'$ then $h(s,q,\eta)=h(s,q,\eta')$;
\item $h(s,q,\eta)=0$ if $(s,q,[\eta]_\sim)$ cannot reach a final vertex in $G$, and $h(s,q,\eta)=1$ if $q\in F$;
\item if $(s,q,[\eta]_\sim)$ can reach a final vertex in $G$ and $q\not\in F$, then
\[
\nabla^+_\mathbf{1}h(s,q,\eta)=\lambda(s)\cdot h(s,q,\eta)-\lambda(s)\cdot\sum_{u\in S}\mathbf{P}(s,u)\cdot h(u,\mathbf{q}^{\eta^+}_{q,s},\eta[\mathbf{X}^{\eta^+}_{q,s}:=0])
\]
and
\[
\nabla^-_\mathbf{1}h(s,q,\eta)=\lambda(s)\cdot h(s,q,\eta)-\lambda(s)\cdot\sum_{u\in S}\mathbf{P}(s,u)\cdot h(u,\mathbf{q}^{\eta^-}_{q,s},\eta[\mathbf{X}^{\eta^-}_{q,s}:=0])
\]
when $\eta(x)>0$ for all $x\in\mathcal{X}$.
\end{enumerate}
\end{theorem}
\begin{proof}
It is clear that $\mathsf{prob}$ satisfies the differential equations above. For the uniqueness, we prove that all functions $h$ that satisfies the differential equations will satisfy the integral equations specified in Theorem~\ref{thm:unique:int}. Let $h$ be such a function. The situation is clear when $(s,q,[\eta]_\sim)$ cannot reach a final vertex in $G$ or $q\in F$. Below we only consider $(s,q,\eta)$ such that $(s,q,[\eta])$ can reach a final vertex in $G$ however $q\not\in F$. For each such $(s,q,\eta)$, we define $f[s,q,\eta]:\mathbb{R}_{\ge 0}\mapsto [0,1]$ by $f[s,q,\eta](t):=h(s,q,\eta+t)$. Then $f[s,q,\eta]$ is differential at those points $t$ where $[\eta+t]_\sim$ is not a marginal region. Note that there are only finitely many points $t$ such that $[\eta+t]_\sim$ is marginal, we can divide $[0,+\infty)$ into a finite cluster $\mathcal{I}$ of open intervals, disregarding a finite number of isolating points, where for each $I\in\mathcal{I}$, $\eta+t\sim\eta+t'$ for all $t',t\in I$. Thus $f[s,q,\eta]$ is piecewise differentiable on $\mathbb{R}_{\ge 0}$. Then $f[s,q,\eta]$ is Lipschitz continuous since the existence and boundness of $\nabla^+_\mathbf{1}h$ and $\nabla^-_\mathbf{1}h$.

Consider $t\in\mathbb{R}_{\ge 0}$ such that $[\eta+t]_\sim$ is not  marginal. We have
\begin{align}\label{eq:2}
&\frac{\mathrm{d}}{\mathrm{d}t}f[s,q,\eta](t)=\\\nonumber
&\qquad\lambda(s)\cdot h(s,q,\eta+t)-\lambda(s)\cdot\sum_{u\in S}\mathbf{P}(s,u)\cdot h(u,\mathbf{q}^{\eta+t}_{q,s},(\eta+t)[\mathbf{X}^{\eta+t}_{q,s}:=0])
\end{align}
Multiply each side of~(\ref{eq:2}) by $e^{-\lambda(s)t}$, we obtain that the equality
\begin{align*}
& e^{-\lambda(s)t}\cdot\frac{\mathrm{d}}{\mathrm{d}t}f[s,q,\eta](t)-\lambda(s)e^{-\lambda(s)t}\cdot f[s,q,\eta](t)=\\
&\qquad{}-\lambda(s)e^{-\lambda(s)t}\cdot\sum_{u\in S}\mathbf{P}(s,u)\cdot h(u,\mathbf{q}^{\eta+t}_{q,s},(\eta+t)[\mathbf{X}^{\eta+t}_{q,s}:=0])
\end{align*}
This is essentially
\begin{align*}
&\frac{\mathrm{d}}{\mathrm{d}t}\left\{e^{-\lambda(s)t}\cdot f[s,q,\eta](t)\right\}=\\
&\qquad{}-\lambda(s)e^{-\lambda(s)t}\cdot\sum_{u\in S}\mathbf{P}(s,u)\cdot h(u,\mathbf{q}^{\eta+t}_{q,s},(\eta+t)[\mathbf{X}^{\eta+t}_{q,s}:=0])
\end{align*}
Note that
$e^{-\lambda(s)t}\cdot f[s,q,\eta](t)$ is absolutely continuous on any closed interval since it is Lipschitz continuous. Thus by the Fundamental Theorem of Calculus for Lebesgue Integral~\cite{Rudin:1987:RCA:26851}, for each $I\in\mathcal{I}$ with $I=(\inf I,\sup I)$ ,
\[
\int_I\frac{\mathrm{d}}{\mathrm{d}t}\left\{e^{-\lambda(s)t}\cdot f[s,q,\eta](t)\right\}\,\mathrm{d}t=e^{-\lambda(s)t}\cdot f[s,q,\eta](t){\Big|_{\inf I}^{\sup I}}
\]
Then by
\begin{align*}
&\int_I\frac{\mathrm{d}}{\mathrm{d}t}\left\{e^{-\lambda(s)t}\cdot f[s,q,\eta](t)\right\}\,\mathrm{d}t=\\
&\qquad{}-\int_I\lambda(s)e^{-\lambda(s)t}\cdot\sum_{u\in S}\mathbf{P}(s,u)\cdot h\left(u,\mathbf{q}^{\eta+t}_{q,s},(\eta+t)[\mathbf{X}^{\eta+t}_{q,s}:=0]\right)\,\mathrm{d}t
\end{align*}
we have
\begin{align*}
&-\int_0^{+\infty}\lambda(s)e^{-\lambda(s)t}\cdot\sum_{u\in S}\mathbf{P}(s,u)\cdot h\left(u,\mathbf{q}^{\eta+t}_{q,s},(\eta+t)[\mathbf{X}^{\eta+t}_{q,s}:=0]\right)\,\mathrm{d}t\\
&\qquad\qquad=\sum_{I\in\mathcal{I}}e^{-\lambda(s)t}\cdot f[s,q,\eta](t){\Big|_{\inf I}^{\sup I}}=-h(s,q,\eta)
\end{align*}
Thus we obtain
\[
h(s,q,\eta)=\int_0^{+\infty}\Lambda_s(t)\cdot\sum_{u\in S}\mathbf{P}(s,u)\cdot h(u,\mathbf{q}^{\eta+t}_{q,s},(\eta+t)[\mathbf{X}^{\eta+t}_{q,s}:=0])\,\mathrm{d}t
\]
Then $h$ satisfies the prerequisite of Theorem~\ref{thm:unique:int}. So $h$ is unique.
\end{proof}

\section{Finite Difference Methods}

In this section, we deal with the approximation of the function $\mathsf{prob}$ through finite approximation schemes. We establish our approximation scheme based on Theorem~\ref{thm:unique:diff} and by finite difference methods~\cite{jw.thomas}. Then we prove that our approximation scheme converges to $\mathsf{prob}$ with a derived error bound.

We fix a CTMC $\mathcal{M}=(S,L,\mathbf{P},\lambda,\mathcal{L})$ and a DTA $\mathcal{A}=(Q,L,\mathcal{X},\Delta,F)$. For computational purpose we assume that all numerical values in $\mathcal{M}$ are rational.

Given valuation $\eta$ and $t\ge 0$, we define $\eta\oplus t\in\mathrm{Val}(\mathcal{X})$ by: $(\eta\oplus t)(x)=\min\{T_x,\eta(x)+t\}$ for all $x\in\mathcal{X}$. Note that $\eta\oplus 0=\eta$ iff $\eta(x)\le T_x$ for all clocks $x$.
We extend $\oplus$, $+$, $\lambda(\centerdot)$, $[\centerdot]_\sim$ and $\mathbf{P}(\centerdot,\centerdot)$ to a triple $(s,q,\eta)\in\mathcal{V}$ as follows:
\begin{itemize} \itemsep0pt \parskip0pt \parsep0pt \topsep0pt
\item $(s,q,\eta)\oplus t=(s,q,\eta\oplus t)$ and $(s,q,\eta)+t=(s,q,\eta+t)$;
\item $[(s,q,\eta)]_\sim=(s,q,[\eta]_\sim)$;
\item $\lambda(s,q,\eta)=\lambda(s)$ and $\mathbf{P}\left((s,q,\eta),u\right)=\mathbf{P}(s,u)$ for $u\in S$.
\end{itemize}
Furthermore, we say that $[(s,q,\eta)]_\sim$ is \emph{marginal} if $[\eta]_\sim$ is marginal. Note that by Lipschitz Continuity and Proposition~\ref{lem:boundness}, we have $\mathsf{prob}(v+t)=\mathsf{prob}(v\oplus t)$ for all $v\in\mathcal{V}$ and $t\ge 0$.

Given $(s,q,\eta)\in\mathcal{V}$ and $u\in S$, we denote $(s,q,\eta)^+:=(s,q,\eta^+)$, and denote the triple $(u,\mathbf{q}^{\eta^+}_{q,s},\eta[\mathbf{X}^{\eta^+}_{q,s}:=0])\in\mathcal{V}$ by $(s,q,\eta)^+_u$.

Before we introduce our approximation scheme, we first prove a useful lemma.

\begin{lemma}\label{lem:region:auxi}
Let $v\in\mathcal{V}$ such that $[v]_\sim$ cannot reach a final vertex in $G$. Then $[v\oplus t]_\sim$ cannot reach a final vertex in $G$ for all $t\ge 0$, and $[v^+_u]_\sim$ cannot reach a final vertex in $G$ for all $u\in S$ such that $\mathbf{P}(v,u)>0$.
\end{lemma}
\begin{proof}
Since $[v]_\sim$ cannot reach a final vertex in $G$, we have $[v+t]_\sim$ cannot reach a final vertex in $G$ for all $t\ge 0$. Note that $v+t$ and $v\oplus t$ differ only at those clocks $x$ whose values in $v+t$ are greater than $T_x$. Then $[v+t]_\sim$ can reach a final vertex in $G$ iff $[v\oplus t]_\sim$ can reach a final vertex in $G$ by the fact that $(v+t)+\tau\sim (v\oplus t)+\tau$ for all $\tau>0$. Thus we have $[v\oplus t]_\sim$ cannot reach a final vertex in $G$.

Let $u\in S$ such that $\mathbf{P}(v,u)>0$. By Lemma~\ref{lem:rightleft}, there exists $t_1>0$ such that $[v+t]_\sim$ does not change for all $t\in (0,t_1)$. Then $[v+t]_\sim$ is not marginal for $t\in (0,t_1)$. Since $[v]_\sim$ cannot reach a final vertex in $G$, we have $[(v+t)_u^+]_\sim$ cannot reach a final vertex in $G$ for $t\in (0,t_1)$. Then by Lemma~\ref{lem:prg}, $\mathsf{prob}((v+t)_u^+)=0$ for all $t\in (0,t_1)$. Thus by Corollary~\ref{crlly:lipschitz}, we have $\mathsf{prob}(v_u^+)=0$. It follows that $[v^+_u]_\sim$ cannot reach a final vertex in $G$.
\end{proof}

\subsection{Approximation Schemes}

We establish our approximation scheme in two steps: firstly, we discretize the hypercube $\prod_{x\in\mathcal{X}}[0,T_x]\subseteq\mathrm{Val}(\mathcal{X})$ into small grids; secondly, we establish our approximation scheme by building constraints between these discrete values through finite difference methods. By Lipschitz Continuity and Proposition~\ref{lem:boundness}, we don't need to consider clock valuations outside $\prod_{x\in\mathcal{X}}[0,T_x]$. The  discretization is as follows.

\begin{definition}
Let $m\in\mathbb{N}$. A clock valuation $\eta$ is \emph{on $m$-grid} if $\eta(x)\in [0, T_x]$ and $\eta(x)\cdot m$ is an integer for all clocks $x$. The set of discrete values $\mathsf{D}_m$ of concern is defined as follows:
\[
\mathsf{D}_m=\{\mathsf{h}[(s,q,\eta)]\mid (s,q,\eta)\in \mathcal{V}\mbox{ and }\eta\mbox{ is on }m\mbox{-grid}\}\enskip.
\]
\end{definition}
Below we fix a $m\in\mathbb{N}$ and define $\rho:=m^{-1}$. Based on Theorem~\ref{thm:unique:diff}, we establish our basic approximation scheme, as follows.

\begin{definition}[Basic Approximation Scheme]\label{def:approx:basic}
The approximation scheme $\Gamma_m$ consists of the discrete values $\mathsf{D}_m$ and a system of linear equations on $\mathsf{D}_m$. The system of linear equations contains one of the following equations for each $\mathsf{h}[v]\in\mathsf{D}_m$:
\begin{itemize} \itemsep0pt \parskip0pt \parsep0pt \topsep0pt
\item $\mathsf{h}[v]=0$ if $[v]_\sim$ (as a vertex of $G$) cannot reach a final vertex in $G$;
\item $\mathsf{h}[v]=1$ if $[v]_\sim$ is a final vertex in $G$;
\item If $[v]_\sim$ can reach a final vertex in $G$ however itself is not a final vertex, then
\[
\frac{\mathsf{h}[v\oplus\rho]-\mathsf{h}[v]}{\rho}=\lambda(v)\cdot\mathsf{h}[v]-\lambda(v)\cdot\sum_{u\in S}\mathbf{P}(v,u)\cdot\mathsf{h}[v^+_u]\enskip.
\]
\end{itemize}
\end{definition}
In other words, we relate elements of $\mathsf{D}_m$ by using $\nabla^+_\mathbf{1}\mathsf{h}$ in Theorem~\ref{thm:unique:diff}. Note that $\mathsf{h}[v]$ is in essence $v$. Sometimes we will not distinguish between $\mathsf{h}[v]$ and $v$.

We note that $\Gamma_m$ does not have initial values from which we can approximate $\mathsf{prob}$ incrementally. One fundamental problem is whether $\Gamma_m$ has a solution, or even a unique solution. Another fundamental problem is the error bound $\max\{|h^*[v]-\mathsf{prob}(v)|\mid \mathsf{h}[v]\in\mathsf{D}_m\}$ provided that $h^*$ is the unique solution of $\Gamma_m$.

Below we first derive the \emph{error bound} of $\Gamma_m$ which is the error bound of each linear equality when we substitute all $\mathsf{h}[v]$ by $\mathsf{prob}(v)$. Note that generally the error bound of an approximation scheme does not imply any information of the error bound of the solution to the approximation scheme.

\begin{proposition}\label{lem:err:scheme}
For all $\mathsf{h}[v]\in\mathsf{D}_m$, if $[v]_\sim$ is not a final vertex and can reach some final vertex in $G$ then $\left|\frac{1}{\rho}\cdot(\mathsf{prob}(v\oplus\rho)-\mathsf{prob}(v))-\nabla_\mathbf{1}^+\mathsf{prob}(v)\right|<M_2\cdot\rho$,
where $M_2:=2\lambda_{\max}\cdot M_1$\enskip.
\end{proposition}
\begin{proof}
Suppose $\mathsf{h}[v]\in\mathsf{D}_m$ such that $[v]_\sim$ is not a final vertex and can reach some final vertex
in $G$. Since $\mathsf{h}[v]\in\mathsf{D}_m$, the function $f[v](t):t\mapsto\mathsf{prob}(v\oplus t)$ is continuous on $[0,\rho]$ and is differentiable on $(0,\rho)$. By Lagrange's Mean Value Theorem, there exists $\rho'\in (0,\rho)$ such that $\frac{1}{\rho}\cdot(\mathsf{prob}(v\oplus\rho)-\mathsf{prob}(v))=\frac{\mathrm{d}}{\mathrm{d}t}f[v](\rho')$.
By Theorem~\ref{thm:drv}, we have
\[
\frac{\mathrm{d}}{\mathrm{d}t}f[v](\rho')=\lambda(v)\cdot\mathsf{prob}(v+\rho')-\lambda(v)\cdot\sum_{u\in S}\mathbf{P}(v,u)\cdot\mathsf{prob}((v+\rho')^+_u)
\]
and
\[
\nabla^+_\mathbf{1}\mathsf{prob}(v)=\lambda(v)\cdot\mathsf{prob}(v)-\lambda(v)\cdot\sum_{u\in S}\mathbf{P}(v,u)\cdot\mathsf{prob}(v^+_u)
\]
Then by Corollary~\ref{crlly:lipschitz}, we obtain the desired result.
\end{proof}
To analyze $\Gamma_m$, we further define several auxiliary approximation schemes. Below we define $\mathsf{B}_m,\mathsf{B}_m^{\max}$ as follows:
\begin{itemize} \itemsep0pt \parskip0pt \parsep0pt \topsep0pt
\item $\mathsf{B}_m=\{\mathsf{h}[v]\in\mathsf{D}_m\mid [v]_\sim\mbox{ is not final and can reach some final vertex in }G\}$
\item $\mathsf{B}_m^{\max}=\{\mathsf{h}[v]\in\mathsf{B}_m\mid v=(s,q,\eta)\mbox{ and }\eta(x)=T_x\mbox{ for all }x\in\mathcal{X}\}$\enskip.
\end{itemize}
For each $\mathsf{h}[v]\in\mathsf{B}_m$, we denote by $N_v\in\mathbb{N}_0$ the minimum number such that either $\mathsf{h}\left[v\oplus (N_v\cdot\rho)\right]\in\mathsf{B}^{\max}_m$ or $[v\oplus (N_v\cdot\rho)]_\sim$ cannot reach some final vertex in $G$. We first transform $\Gamma_m$ into an equivalent form.

\begin{definition}
The approximation scheme $\Gamma'_m$ consists of the discrete values $\mathsf{D}_m$, and the system of linear equations which contains one of the following linear equalities for each $\mathsf{h}[v]\in\mathsf{D}_m$:
\begin{itemize} \itemsep0pt \parskip0pt \parsep0pt \topsep0pt
\item $\mathsf{h}[v]=0$ if $[v]_\sim$ cannot reach a final vertex in ${G}$;
\item $\mathsf{h}[v]=1$ if $[v]_\sim$ is a final vertex of $G$.
\item If $\mathsf{h}[v]\in\mathsf{B}_m\backslash\mathsf{B}_m^{\max}$, then
\begin{equation}\label{eq:3}
\mathsf{h}[v]=\frac{1}{1+\rho\cdot\lambda(v)}\cdot
\mathsf{h}[v\oplus\rho]+\frac{\rho\cdot\lambda(v)}{1+\rho\cdot\lambda(v)}\cdot\sum_{u\in S}\mathbf{P}(v,u)\cdot\mathsf{h}[v^+_u]
\end{equation}
\item if $\mathsf{h}[v]\in\mathsf{B}^{\max}_m$ then $\mathsf{h}[v]=\sum_{u\in S}\mathbf{P}(v,u)\cdot\mathsf{h}[v^+_u]$\enskip.
\end{itemize}
\end{definition}
It is clear that $\Gamma'_m$ is just a re-formulation of $\Gamma_m$. Note that the case $\mathsf{h}[v]\in\mathsf{B}^{\max}_m$ is derived from $v\oplus\rho=v$. The error bound of $\Gamma'_m$ is as follows.

\begin{proposition}\label{lem:err:gamma1}
For all $\mathsf{h}[v]\in\mathsf{B}^{\max}_m$, $\mathsf{prob}(v)=\sum_{u\in S}\mathbf{P}(v,u)\cdot\mathsf{prob}(v^+_u)$. For all $\mathsf{h}[v]\in\mathsf{B}_m\backslash\mathsf{B}_m^{\max}$,
\[
\left|\mathsf{prob}(v)-\left(\frac{1}{1+\rho\cdot\lambda(v)}\cdot\mathsf{prob}(v\oplus\rho)+\frac{\rho\cdot\lambda(v)}{1+\rho\cdot\lambda(v)}\cdot\sum_{u\in S}\mathbf{P}(v,u)\cdot\mathsf{prob}\left(v^+_u\right)\right)\right|<M_2\rho^2\enskip.
\]
\end{proposition}
\begin{proof}
The case $\mathsf{h}[v]\in\mathsf{B}^{\max}_m$ is due to the fact that $\nabla^+_\mathbf{1}\mathsf{prob}(v)=0$. The case $\mathsf{h}[v]\in\mathsf{B}_m\backslash\mathsf{B}_m^{\max}$ can be directly derived from the statement of Proposition~\ref{lem:err:scheme}.
\end{proof}

\begin{remark}
$\Gamma'_m$ can somewhat be viewed as a reachability problem on a discrete-time Markov chain (DTMC). The nodes of the DTMC are elements of $\mathsf{D}_m$. The goal nodes are those $\mathsf{h}[v]$'s where $[v]_\sim$ is a final vertex of $G$. The set of nodes which can not reach a goal state are those $\mathsf{h}[v]$'s where $[v]_\sim$ cannot reach a final vertex in $G$. Then $\Gamma'_m$ states the relationship between the reachability probabilities of the remaining states. It seems that we can use this fact to deduce that $\Gamma'_m$ (hence $\Gamma_m$) has a unique solution (e.g., by applying the proof on~\cite[Page 766]{DBLP:books/daglib/0020348}). However, this may fail because the remaining states may still contain states which cannot reach goal states. This can be seen as follows.

Suppose the DTA has four locations $\{q_0,q_1,q_2,q_3\}$ where $q_2$ is the only final location and two clocks $\{x,y\}$. Let $s$ be a CTMC-state with $\mathbf{P}(s,s)>0$. Assume that the DTA has only two meaningful rules, namely $(q_0,\mathcal{L}(s),x<1\wedge y<1,\{x\},q_1)$ and $(q_1,\mathcal{L}(s),x<1\wedge y>1,\emptyset,q_2)$. The other rules lead to the "deadlock" location $q_3$. Define $\eta=(\frac{m-1}{m},0)$ where the first (resp. second) coordinate is the value on $x$ (resp. on $y$). Then $(s,q_0,[\eta]_\sim)$ can reach a final vertex in $G$ since we have $(q_0,\eta)\xrightarrow{\mathcal{L}(s),\frac{1}{2}\cdot\rho}(q_1,(0,\frac{1}{2}\cdot\rho))$ and $(q_1,(0,\frac{1}{2}\cdot\rho))\xrightarrow{\mathcal{L}(s),1-\frac{1}{4}\cdot\rho}(q_2,(1-\frac{1}{4}\cdot\rho,1+\frac{1}{4}\cdot\rho))$. However after the discretization, the node $(s,q_0,\eta)$ can go to $(s,q_1,(0,0))$ but from $(s,q_1,(0,0))$ we cannot go to $q_2$. This implies that $(s,q_0,\eta)$ cannot reach a final node.\qed
\end{remark}
Below we unfold $\Gamma'_m$ into another equivalent form $\Gamma''_m$.

\begin{definition}
The approximation scheme $\Gamma''_m$ consists of the discrete values $\mathsf{D}_m$, and one of the following linear equality for each $\mathsf{h}[v]\in\mathsf{D}_m$:
\begin{itemize}\itemsep0pt \parskip0pt \parsep0pt \topsep0pt
\item $\mathsf{h}[v]=0$ if $[v]_\sim$ cannot reach some final vertex in $G$, and $\mathsf{h}[v]=1$ if $[v]_\sim$ is a final vertex;
\item if $\mathsf{h}[v]\in\mathsf{B}_m\backslash\mathsf{B}^{\max}_m$, then
\begin{align}\label{eq:4}
\mathsf{h}[v]&=\sum_{l=0}^{N_v-1}\left\{\left(\frac{1}{1+\rho\cdot\lambda(v)}\right)^l\cdot\frac{\rho\cdot\lambda(v)}{1+\rho\cdot\lambda(v)}\cdot\sum_{u\in S}\mathbf{P}(v,u)\cdot\mathsf{h}\left[(v\oplus(l\cdot\rho))^+_u\right]\right\}\nonumber\\
& \quad{}+\left(\frac{1}{1+\rho\cdot\lambda(v)}\right)^{N_v}\cdot f(v)
\end{align}
where $f(v):=0$ if $[v\oplus (N_v\cdot\rho)]_\sim$ cannot reach some final vertex in $G$ and $f(v):=\sum_{u\in S}\mathbf{P}(v,u)\cdot\mathsf{h}[(v\oplus (N_v\cdot\rho))^+_u]$ if $\mathsf{h}\left[v\oplus (N_v\cdot\rho)\right]\in\mathsf{B}^{\max}_m$;
\item if $\mathsf{h}[v]\in\mathsf{B}^{\max}_m$ then $\mathsf{h}[v]=\sum_{u\in S}\mathbf{P}(v,u)\cdot\mathsf{h}[v^+_u]$\enskip.
\end{itemize}
\end{definition}
Intuitively, $\Gamma''_m$ is obtained by unfolding $\mathsf{h}[v\oplus\rho]$ further in Equation~(\ref{eq:3}). In the following, we prove that $\Gamma'_m$ and $\Gamma''_m$ are equivalent.

We describe $\Gamma''_m$ by a matrix equation $\mu=\mathbf{A}\mu+\mathbf{b}$ where $\mu$ is the vector over $\mathsf{B}_m$ to be solved, $\mathbf{b}:\mathsf{B}_m\mapsto\mathbb{R}$ is a vector and $\mathbf{A}:\mathsf{B}_m\times \mathsf{B}_m\mapsto\mathbb{R}$ is a matrix. More specifically, the row $\mathbf{A}(\mathsf{h}[v],-)$ is specified by the coefficents on $\mathsf{h}[v']\in\mathsf{B}_m$ in Equation~(\ref{eq:4}); the value $\mathbf{b}(\mathsf{h}[v])$ is the sum of the values over $\mathsf{D}_m\backslash\mathsf{B}_m$ in Equation~(\ref{eq:4}). The exact permutation among $\mathsf{B}_m$ is irrelevant. Analogously, we describe $\Gamma'_m$ by a matrix equation $\mu=\mathbf{C}\mu+\mathbf{d}$.

\begin{proposition}\label{lem:approx:equiv}
$\Gamma'_m$ and $\Gamma''_m$ are equivalent, i.e., they have the same set of solutions.
\end{proposition}
\begin{proof}
``$\Gamma'_m\Rightarrow\Gamma''_m$'': It is clear that $\Gamma''_m$ is obtained directly from $\Gamma'_m$ by expanding $\mathsf{h}[v\oplus\rho]$ iteratively in Equation~(\ref{eq:3}) whenever $v\oplus\rho\in\mathsf{B}_m\backslash\mathsf{B}_m^{\max}$.

``$\Gamma''_m\Rightarrow\Gamma'_m$'': Let $\{h[v]\mid\mathsf{h}[v]\in\mathsf{D}_m\}$ be a solution of $\Gamma''_m$. We define $\hat{h}:\mathsf{D}_m\mapsto\mathbb{R}$ by $\hat{h}[v]=h[v]\cdot{\langle}{\mathsf{B}_m}{\rangle}(v)$ for $\mathsf{h}[v]\in\mathsf{D}_m$. We prove that for all $\mathsf{h}[v]\in\mathsf{B}_m\backslash\mathsf{B}_m^{\max}$ and all $0\le m< N_v$,
\begin{align}\label{eq:5}
h[v]&=\sum_{l=0}^{m}\left\{\left(\frac{1}{1+\rho\cdot\lambda(v)}\right)^l\cdot\frac{\rho\cdot\lambda(v)}{1+\rho\cdot\lambda(v)}\cdot\sum_{u\in S}\mathbf{P}(v,u)\cdot \hat{h}\left[(v\oplus(l\cdot\rho))^+_u\right]\right\}\nonumber\\
&\quad+\sum_{l=0}^{m}\left\{\left(\frac{1}{1+\rho\cdot\lambda(v)}\right)^l\cdot\mathbf{d}\left(\mathsf{h}\left[v\oplus(l\cdot\rho)\right]\right)\right\}\nonumber\\
& \quad+\left(\frac{1}{1+\rho\cdot\lambda(v)}\right)^{m+1}\cdot h\left[v\oplus ((m+1)\cdot\rho)\right]\enskip.
\end{align}
We prove this by induction on $N_v-m$. The case when $m=N_v-1$ is directly specified by $\Gamma''_m$. Let $0\le m<m+1<N_v$ and suppose Equation~\ref{eq:5} holds when $N_{v'}-m'<N_v-m$. Then we have
\begin{align*}
h[v]&=\sum_{l=0}^{m+1}\left\{\left(\frac{1}{1+\rho\cdot\lambda(v)}\right)^l\cdot\frac{\rho\cdot\lambda(v)}{1+\rho\cdot\lambda(v)}\cdot\sum_{u\in S}\mathbf{P}(v,u)\cdot \hat{h}\left[(v\oplus(l\cdot\rho))^+_u\right]\right\}\\
&\quad+\sum_{l=0}^{m+1}\left\{\left(\frac{1}{1+\rho\cdot\lambda(v)}\right)^l\cdot\mathbf{d}\left(\mathsf{h}\left[v\oplus(l\cdot\rho)\right]\right)\right\}\\
& \quad{}+\left(\frac{1}{1+\rho\cdot\lambda(v)}\right)^{m+2}\cdot h\left[v\oplus ((m+2)\cdot\rho)\right]\enskip.
\end{align*}
Then we have $(\dag)$:
\begin{align*}
h[v]&=\sum_{l=0}^{m}\left\{\left(\frac{1}{1+\rho\cdot\lambda(v)}\right)^l\cdot\frac{\rho\cdot\lambda(v)}{1+\rho\cdot\lambda(v)}\cdot\sum_{u\in S}\mathbf{P}(v,u)\cdot \hat{h}\left[(v\oplus l\cdot\rho)^+_u\right]\right\}\\
&\quad{}+\sum_{l=0}^{m}\left\{\left(\frac{1}{1+\rho\cdot\lambda(v)}\right)^l\cdot\mathbf{d}\left(\mathsf{h}\left[v\oplus(l\cdot\rho)\right]\right)\right\}\\
& \quad{}+\left(\frac{1}{1+\rho\cdot\lambda(v)}\right)^{m+1}\cdot\Big\{
\frac{\rho\cdot\lambda(v)}{1+\rho\cdot\lambda(v)}\cdot
\sum_{u\in S}\mathbf{P}(v,u)\cdot \hat{h}\left[(v\oplus ((m+1)\cdot\rho))^+_u\right]\\
& \quad{}+\mathbf{d}\left(\mathsf{h}\left[v\oplus((m+1)\cdot\rho)\right]\right)+\left(\frac{1}{1+\rho\cdot\lambda(v)}\right)\cdot h\left[v\oplus ((m+2)\cdot\rho)\right]\Big\}\enskip.
\end{align*}
Note that $N_{v\oplus(m+1)\cdot\rho}-0=N_v-(m+1)<N_v-m$. Thus by the induction hypothesis,
\begin{align*}
& h[v\oplus((m+1)\cdot\rho)]\\
= & \frac{\rho\cdot\lambda(v)}{1+\rho\cdot\lambda(v)}\cdot
\sum_{u\in S}\mathbf{P}(v,u)\cdot \hat{h}\left[(v\oplus((m+1)\cdot\rho))_u^+\right]\\
&{}+\mathbf{d}\left(\mathsf{h}\left[v\oplus((m+1)\cdot\rho)\right]\right)+\left(\frac{1}{1+\rho\cdot\lambda(v)}\right)\cdot h\left[v\oplus ((m+2)\cdot\rho)\right]\enskip.
\end{align*}
Thus we obtain that Equation~(\ref{eq:5}) holds when we substitute $h[v\oplus((m+1)\cdot\rho)]$ into $(\dag)$. Then by taking $m=0$ in Equation~(\ref{eq:5}), we obtain that $\{h[v]\mid\mathsf{h}[v]\in\mathsf{D}_m\}$ is a solution of $\Gamma'_m$.
\end{proof}
Below we derive the error bound of $\Gamma''_m$.

\begin{proposition}\label{lem:err:gamma2}
The error bound of $\Gamma''_m$ is $M_3\cdot\rho$, where $M_3=T_{\max}\cdot M_2$.
\end{proposition}
\begin{proof}
We only need to consider $\mathsf{h}[v]\in\mathsf{B}_m\backslash\mathsf{B}^{\max}_m$. By Proposition~\ref{lem:err:gamma1},
\begin{equation}\label{eq:6}
\left|\mathsf{prob}(v)-\left(\frac{1}{1+\rho\cdot\lambda(v)}\cdot
\mathsf{prob}(v\oplus\rho)+\frac{\rho\cdot\lambda(v)}{1+\rho\cdot\lambda(v)}\cdot\sum_{u\in S}\mathbf{P}(v,u)\cdot\mathsf{prob}\left(v^+_u\right)\right)\right|<M_2\rho^2
\end{equation}
Expanding $\mathsf{prob}[v\oplus\rho]$ one step further in~(\ref{eq:6}) will result in another error of
$\frac{1}{1+\rho\cdot\lambda(v)}\cdot M_2\rho^2$. By iterated expansion up to $N_v$ $(\le T_{\max}\rho^{-1})$ steps, we obtain that the error bound of $\Gamma''_m$ is no greater than $M_2\rho^2\cdot\sum_{n=0}^{N_v-1}(\frac{1}{1+\rho\cdot\lambda(s)})^n$, which is smaller than $M_2\cdot T_{\max}\cdot\rho$.
\end{proof}

\subsection{Analysis of the Approximation Schemes}

Below we analyse the approximation schemes proposed in the previous subsection. We fix some $m\in\mathbb{N}$ and $\rho=m^{-1}$. We define $\lambda_{\min}=\min\{\lambda(s)\mid s\in S\}$ and $p_{\min}=\min\{\mathbf{P}(s,u)\mid s,u\in S\mbox{ and }\mathbf{P}(s,u)>0\}$. Note that $\lambda_{\min}>0$.

Recall that we describe $\Gamma''_m$ by $\mu=\mathbf{A}\mu+\mathbf{b}$ and $\Gamma'_m$ by $\mu=\mathbf{C}\mu+\mathbf{d}$ in the previous subsection. Below we analyse the equation $\mu=\mathbf{A}\mu+\mathbf{b}$. To this end, we first reproduce (on CTMC and DTA) the notions of $\delta$-\emph{seperateness} and $\delta$-\emph{wideness}, which is originally discovered by Brazdil~\emph{et al.}~\cite{DBLP:conf/hybrid/BrazdilKKKR11} on semi-Markov processes and DTA. Below we define the transition relation $\xrightarrow{t}$ over $\mathcal{V}$ by: $(s,q,\eta)\xrightarrow{t}(u,q',\eta')$ iff $\mathbf{P}(s,u)>0$ and $(q,\eta)\xrightarrow{\mathcal{L}(s),t}(q',\eta')$.

\begin{definition}
A clock valuation $\eta$ is $\delta$-\emph{separated} if for all $d_1,d_2\in\mathcal{R}_\eta$, either $d_1=d_2$ or $|d_1-d_2|\ge\delta$. A transition $(s,q,\eta)\xrightarrow{t} (u,q',\eta')$ is $\delta$-\emph{wide} if $t\ge\delta$ and for all $\tau\in (t-\delta,t+\delta)$, $\eta+\tau\sim\eta+t$. Furthermore, a transition path
\[
(s_0,q_0,\eta_0)\xrightarrow{t_1}(s_1,q_1,\eta_1)\dots\xrightarrow{t_n}(s_n,q_n,\eta_n)\quad (n\ge 1)
\]
is $\delta$-\emph{wide} if all its transitions are $\delta$-wide.
\end{definition}
Intuitively, A transition is $\delta$-wide if one can adjust the transition by up to $\delta$ time units, while keeping the DTA rule used on this transition.

Below we say that a set $\mathcal{I}$ of \emph{disjoint} open intervals is an \emph{open partition} (of $[0,1]$) iff it holds that $\bigcup\mathcal{I}\subseteq[0,1]$ and $[0,1]\backslash\bigcup\mathcal{I}$ is a finite set. Given a non-empty open interval $I\subseteq [0,1]$ with $I=(c_1,c_2)$ and a $t\in\mathbb{R}_{\ge 0}$, we denote by $I\diamond t$ the (possibly empty) interval $(\mathrm{frac}(c_1+t),\mathrm{frac}(c_2+t))$\enskip.
The following result is the counterpart of the one on semi Markov processes and DTA~\cite{DBLP:conf/concur/BrazdilKKKR10,DBLP:conf/hybrid/BrazdilKKKR11}. 

\begin{proposition}\label{lem:deltawide}
For all $(s,q,\eta)\in\mathcal{V}$, if $\eta$ is $\delta$-separated and $(s,q,[\eta]_\sim)$ can reach some final vertex in $G$ with $q\not\in F$, then there exists an at most $|V|$-long, $\delta\slash|V|$-wide transition path from $(s,q,\eta)$ to some $(s',q',\eta')$ with $q'\in F$.
\end{proposition}
\begin{proof}
Let $(s,q,\eta)\in\mathcal{V}$ such that $(s,q,[\eta]_\sim)$ is not final and can reach some final vertex in $G$. Then there is a path
\[
(s,q,[\eta]_\sim)=(s_1,q_1,r_1)\rightarrow\dots\rightarrow (s_n,q_n,r_n)
\]
in $G$ such that $1<n\le |V|$ and $q_n\in F$. We first inductively construct a transition path
\[
(s_1,q_1,\eta_1)\xrightarrow{t_1}\dots\xrightarrow{t_{n-1}}(s_n,q_n,\eta_n)
\]
such that $\eta_i\in r_i$ for all $1\le i\le n$, while maintaining the following structures:
\begin{itemize} \itemsep0pt \parskip0pt \parsep0pt \topsep0pt
\item two open partitions $\mathcal{I}'_i,\mathcal{I}_i$ with $\mathcal{R}_{\eta_i}\subseteq [0,1]\backslash\bigcup\mathcal{I}'_i$ for each $1\le i\le n$;
\item a bijection $\varphi_i:\mathcal{I}'_i\mapsto\mathcal{I}_i$ for each $1\le i\le n$;
\item two intervals $(\mathsf{c}_1^i,\mathsf{c}_2^i)\in\mathcal{I}_i$, $(\mathsf{w}_1^i,\mathsf{w}_2^i)\in\mathcal{I}'_i$ with $\varphi_i((\mathsf{w}_1^i,\mathsf{w}_2^i))=(\mathsf{c}_1^i,\mathsf{c}_2^i)$ for each $1\le i\le n-1$;
\item a $\mathsf{c}^i\in (\mathsf{c}_1^i,\mathsf{c}_2^i)$ for each $1\le i\le n-1$.
\end{itemize}
Initially, we set $\eta_1=\eta$ and $\mathcal{I}_1=\mathcal{I}'_1=\{(w_j,w_{j+1})\mid 0\le j<m\}$, where $\{w_j\}_{0\le j\le m}$ satisfies that $\mathcal{R}_\eta=\{w_0,w_1,\dots,w_m\}$ and $w_j<w_{j+1}$ for all $0\le j<m\enskip$ (note that $w_0=0$ and $w_{m}=1$). We let $\varphi_1$ be the identity mapping.

Suppose that the transition path until $(s_i,q_i,\eta_i)$ together with $\mathcal{I}'_i,\mathcal{I}_i,\varphi_i$ are constructed. Since $(s_i,q_i,r_i)\rightarrow(s_{i+1},q_{i+1},r_{i+1})$, there exists $t_i>0$ such that $[\eta_i+t_i]_\sim$ is not marginal, $q_{i+1}=\mathbf{q}_{q_i,s_i}^{\eta_i+t_i}$ and $(\eta_i+t_i)[\mathbf{X}^{\eta_i+t_i}_{q_i,s_i}:=0]\in r_{i+1}$. Since  $[\eta_i+t_i]_\sim$ is not marginal, we can assume $1\in (\mathsf{w}^i_1+\mathrm{frac}(t_i),\mathsf{w}^i_2+\mathrm{frac}(t_i))$ for some $(\mathsf{w}^i_1,\mathsf{w}^i_2)\in\mathcal{I}'_i$.
Denote $(\mathsf{c}^i_1,\mathsf{c}^i_2)=\varphi_i((\mathsf{w}^i_1,\mathsf{w}^i_2))$ and choose  $\mathsf{c}^i\in (\mathsf{c}^i_1,\mathsf{c}^i_2)$ arbitrarily (e.g., $\mathsf{c}^i=\frac{1}{2}\cdot(\mathsf{c}^i_1+\mathsf{c}^i_2)$). Then we set $\eta_{i+1}:=(\eta_i+t_i)[\mathbf{X}^{\eta_i+t_i}_{q_i,s_i}:=0]\in r_{i+1}$, and split $\mathcal{I}_{i},\mathcal{I}'_{i}$ as follows:
\begin{align*}
\mathcal{I}_{i+1}:=&\varphi_i(\mathcal{I}'_i-\{(\mathsf{w}^i_1,\mathsf{w}^i_2)\})\cup\{(\mathsf{c}^i_1,\mathsf{c}^i),(\mathsf{c}^i,\mathsf{c}^i_2)\}\enskip;\\
\mathcal{I}'_{i+1}:=&\{(\mathsf{w}^i_1+\mathrm{frac}(t_i),1),(0, \mathrm{frac}(\mathsf{w}^i_2+\mathrm{frac}(t_i)))\}\cup\{I\diamond t_i\mid I\in\mathcal{I}'_{i}-\{(\mathsf{w}^i_1,\mathsf{w}^i_2)\}\}\enskip.
\end{align*}
The mapping $\varphi_{i+1}:\mathcal{I}'_{i+1}\mapsto\mathcal{I}_{i+1}$ is defined as follows:
\begin{align*}
&\varphi_{i+1}((\mathsf{w}^i_1+\mathrm{frac}(t_i),1))=(\mathsf{c}^i_1,\mathsf{c}^i)\mbox{ and }\varphi_{i+1}((0,\mathrm{frac}(\mathsf{w}^i_2+\mathrm{frac}(t_i)))=(\mathsf{c}^i,\mathsf{c}^i_2)\enskip;\\
&\varphi_{i+1}(I\diamond t_i)=\varphi_i(I)\mbox{ for all }I\in\mathcal{I}'_{i}-\{(\mathsf{w}^i_1,\mathsf{w}^i_2)\}\enskip.
\end{align*}
Intuitively, we record by $\mathcal{I}'_i$ every splitting point which may make $\eta_i+t_i$ less wide, and we record the splitting information without time delaying by $\mathcal{I}_i$, where the correspondence between them is maintained by $\varphi_i$.

Since $n\le |V|$, at most $|V|-1$ splitting occurs at each interval (including its sub-intervals) in $\mathcal{I}_1$ during the inductive construction above. Based on this point, we inductively construct a $\delta\slash|V|$-wide transition path
\[
(s,q,\eta)=(s_1,q_1,\eta'_1)\xrightarrow{t'_1}\dots\xrightarrow{t'_{n-1}}(s_n,q_n,\eta'_n)
\]
such that $\eta'_i\sim\eta_i$ for all $1\le i\le n$, while maintaining an open partition $\mathcal{I}''_i$ and a bijection $\psi_i:\mathcal{I}''_i\mapsto\mathcal{I}'_i$ for each $1\le i\le n$ which satisfy the following condition for all $1\le i\le n$:
\begin{enumerate} \itemsep0pt \parskip0pt \parsep0pt \topsep0pt
\item $\mathcal{R}_{\eta'_i}\subseteq [0,1]\backslash\bigcup\mathcal{I}''_i$;
\item for all $I_1,I_2\in\mathcal{I}''_i$, $\sup I_1\le \inf I_2$ iff $\sup\psi_i(I_1)\le\inf\psi_i(I_2)$;
\item for all $x\in\mathcal{X}$, $(c'_1,c'_2)\in\mathcal{I}''_i$ and $(c_1,c_2)\in\mathcal{I}'_i$, if $\psi_i((c'_1,c'_2))=(c_1,c_2)$, then (i) $\eta'_i(x)=c'_1$ iff $\eta_i(x)=c_1$ and (ii) $\eta'_i(x)=c'_2$ iff $\eta_i(x)=c_2$.
\end{enumerate}
Intuitively, we maintain the order on the fractional values in the previous transition path, however adjust the gap on each transition.

Initially we set $\mathcal{I}''_1=\mathcal{I}'_1$ and $\psi_1$ to be the identity mapping. Suppose the path until $(s_i,q_i,\eta'_i)$ is constructed. Let $(\mathsf{d}^i_1,\mathsf{d}^i_{2})\in\mathcal{I}''_i$ such that $\psi_i((\mathsf{d}^i_{1},\mathsf{d}^i_{2}))=(\mathsf{w}^i_{1},\mathsf{w}^i_{2})$. We choose $t'_i$ such that $\mathrm{int}(t'_i)=\mathrm{int}(t_i)$, $1\in (\mathsf{d}^i_{1}+\mathrm{frac}(t'_i),\mathsf{d}^i_{2}+\mathrm{frac}(t'_i))$ and the length of  $(\mathsf{d}^i_{1}+\mathrm{frac}(t'_i),1)$ (resp. $(0, \mathrm{frac}(\mathsf{d}^i_{2}+\mathrm{frac}(t'_i)))$ is no smaller than $(k^i_1+1)\cdot\delta\slash|V|$ (resp. $(k^i_2+1)\cdot\delta\slash|V|$), where $k^i_1$ (resp. $k^i_2$) is the number of splittings on the sub-intervals of $(\mathsf{c}_1^i,\mathsf{c}^i)$ (resp. $(\mathsf{c}^i,\mathsf{c}_2^i)$) during the previous inductive construction. Then we set $\eta'_{i+1}:=(\eta'_i+t'_i)[\mathbf{X}_{q_i,s_i}^{\eta'_i+t'_i}:=0]$ and similarly split $\mathcal{I}''_i,\psi_i$ as follows:
\begin{align*}
&\mathcal{I}''_{i+1}:=\{I\diamond t'_i\mid I\in\mathcal{I}''_{i}-\{(\mathsf{d}^i_1,\mathsf{d}^i_2)\}\}\cup\{(\mathsf{d}^i_{1}+\mathrm{frac}(t'_i), 1),(0, \mathrm{frac}(\mathsf{d}^i_{2}+\mathrm{frac}(t'_i)))\};\\
&\psi_{i+1}(I'\diamond t'_i)=I\diamond t_i\mbox{ whenever }\psi_{i}(I')=I\mbox{ and }I'\ne (\mathsf{d}^i_1,\mathsf{d}^i_2);\\
&\psi_{i+1}((\mathsf{d}^i_{1}+\mathrm{frac}(t'_i), 1))=(\mathsf{w}^i_{1}+\mathrm{frac}(t_i), 1);\\
&\psi_{i+1}((0, \mathrm{frac}(\mathsf{d}^i_{2}+\mathrm{frac}(t'_i))))=(0, \mathrm{frac}(\mathsf{w}^i_{2}+\mathrm{frac}(t_i)))\enskip.
\end{align*}
By the choice of $t'_i,t_i$, one can prove inductively that $\eta_i\sim\eta'_i$ and $\eta_i+t_i\sim\eta'_i+t'_i$ for all $1\le i\le n$. Thus the constructed path is a legal transition path. And by the construction, this transition path is $\delta\slash|V|$-wide.
\end{proof}
Then we study the linear equation system $\mu=\mathbf{A}\mu+\zeta$ where $\zeta:\mathsf{B}_m\mapsto\mathbb{R}$ is an arbitrary real vector. Below we define $\zeta_{\max}:\mathsf{B}_m\mapsto\mathbb{R}$ such that $\zeta_{\max}(\mathsf{h}[v])=M_3\cdot\rho$ for all $\mathsf{h}[v]\in\mathsf{B}_m$. We denote by $|\zeta|$ the vector such that $|\zeta|(\mathsf{h}[v])=|\zeta(\mathsf{h}[v])|$ for all $\mathsf{h}[v]\in\mathsf{B}_m$. We extend $\le$ to vectors over $\mathsf{B}_m$ in a pointwise fashion.

\begin{proposition}~\label{lem:leq:existence}
Suppose $m>2|V|^2$. Let $\zeta$ be a vector over $\mathsf{B}_m$ such that $|\zeta|\le\zeta_{\max}$. Then the matrix series $\sum_{n=0}^\infty\mathbf{A}^n\zeta$ converges. Moreover, we have $\parallel\sum_{n=0}^\infty \mathbf{A}^n\zeta\parallel_\infty\le |V|\cdot\mathfrak{c}^{-|V|}\cdot(M_3\cdot\rho)$, where $\mathfrak{c}:=e^{-\lambda_{\max}T_{\max}}\cdot p_{\min}\cdot \frac{\lambda_{\min}}{2|V|^2+\lambda_{\min}}$\enskip.
\end{proposition}
\begin{proof}
Let $\delta:=|V|^{-2}$ and $k:={\lfloor}{m\slash|V|^2}{\rfloor}$. We analyse $(\sum_{n=0}^\infty\mathbf{A}^n\zeta)(\mathsf{h}[v^*])$ for each $\mathsf{h}[v^*]\in\mathsf{B}_m$. Firstly, we consider the case when $\zeta=\zeta_{\max}$ and $\mathsf{h}[v^*]\in\mathsf{B}^{\max}_m$. Denote $v^*=(s^*,q^*,\eta^*)$. By definition, $\eta^*$ is $1$-separated. Then by Proposition~\ref{lem:deltawide}, there exists a shortest $|V|^{-1}$-wide path
\[
(s^*,q^*,\eta^*)=(s_1,q_1,\eta_1)\xrightarrow{t_1}\dots\xrightarrow{t_{n-1}}(s_n,q_n,\eta_n)
\]
with $1<n\le |V|$, $q_i\not\in F$ for $1\le i\le n-1$ and $q_n\in F$. Note that $[\eta_i+t_i]_\sim$ is not marginal for $1\le i\le n-1$. We adjust the delay-times in the transition path up to $\delta$ by:
\[
(s^*,q^*,\eta^*)=(s_1,q'_1,\eta'_1)\xrightarrow{t_1+\delta_1}\dots\xrightarrow{t_{n-1}+\delta_{n-1}}(s_n,q'_n,\eta'_n)
\]
where $\delta_i\in [0,\delta)$ for all $1\le i\le n-1$. Given arbitrary $\{\delta_i\}_i$, one can prove inductively that $(\dag)$:
\[
q'_i=q_i, \eta'_i\equiv_\mathrm{g}\eta_i\mbox{ and }\eta'_i+(t_i+\delta_i)\equiv_\mathrm{g}\eta_i+t_i\mbox{ for all }{1\le i\le n-1}\enskip,
\]
\noindent by checking that: $\eta'_i(x)+(t_i+\delta_i)\le\eta_i(x)+t_i+\sum_{j=1}^{i}\delta_j<\eta_i(x)+t_i+|V|^{-1}$ for all clocks $x$, at each stage of induction.

Define $\{\mathsf{V}_i\}_{1\le i\le n}$ with each $\mathsf{V}_i\subseteq\mathcal{V}$ as follows:
\[
\mathsf{V}_1=\{v^*\}\mbox{ and }\mathsf{V}_{i+1}=\{(v+\tau)^+_{s_{i+1}}\mid v\in\mathsf{V}_i, \tau\in [t_i,t_i+\delta)\}\enskip.
\]
By $(\dag)$, we have $[v]_\sim$ is not final and can reach some final vertex in $G$ for all $v\in\bigcup_{i=1}^{n-1}\mathsf{V}_i$. Let $\mathsf{V}'_i:=\{v\oplus 0\mid v\in\mathsf{V}_i\}$ for each $1\le i\le n$. Note that $v\oplus 0$ and $v$ only differ at clocks $x$ whose value is greater than $T_x$. Thus we have $[v]_\sim$ can reach some final vertex in $G$ iff $[v\oplus 0]_\sim$ can reach some final vertex in $G$, for any $v\in\mathcal{V}$. It follows that  $[v]_\sim$ is not final vertex and can reach some final vertex in $G$ for all $v\in\bigcup_{i=1}^{n-1}\mathsf{V}'_i$.

Then we define $\{\mathsf{B}'_i\}_{1\le i\le n}$ with each $\mathsf{B}'_i\subseteq\mathsf{D}_m$ by:
\[
\mathsf{B}'_1=\{\mathsf{h}[v^*]\}\mbox{ and }\mathsf{B}'_{i+1}=\bigcup\{\mathrm{Post}_i(v)\mid \mathsf{h}[v]\in\mathsf{B}'_i\},
\]
where for each $v\in\mathcal{V}$:
\begin{itemize} \itemsep0pt \parskip0pt \parsep0pt \topsep0pt
\item $\mathrm{Post}_i(v):=\{\mathsf{h}[(v')^+_{s_{i+1}}]\mid v'\in\mathrm{Delay}_i(v)\}$ ;
\item $\mathrm{Delay}_i(v):=\{v\oplus\tau\mid\tau\in [t_i,t_i+\delta), \mathsf{h}[v\oplus\tau]\in\mathsf{D}_m\}$ .
\end{itemize}
Consider any $v\in\mathsf{V}_i$ and $\tau\in[t_i,t_i+\delta)$ with $1\le i\le n-1$.
By $(\dag)$, we have $[v+\tau]_\sim$ is not marginal. Then $[(v+\tau)^+]_\sim=[(v\oplus\tau)^+]_\sim$. It follows that $(v+\tau)^+_{s_{i+1}}\oplus 0=(v\oplus\tau)^+_{s_{i+1}}$. Then one can prove inductively that $\mathsf{B}'_i\subseteq\mathsf{V}'_i$ for all $1\le i\le n$. It follows that  $\bigcup_{i=1}^{n-1}\mathsf{B}'_i\subseteq\mathsf{B}_m$.

We prove by induction on $i\ge 1$ that for all $v\in\mathsf{B}'_{n-i}$, $|(\mathbf{A}^i\zeta_{\max})(v)|\le (1-\mathfrak{c}^i)\cdot M_3\rho$. Note that $\mathbf{A}\zeta_{\max}\le\zeta_{\max}$ and $\mathbf{A}\zeta_1\le\mathbf{A}\zeta_2$ for all $\vec{0}\le\zeta_1\le\zeta_2$. It follows that $\mathbf{A}^j\zeta_{\max}\le\mathbf{A}^i\zeta_{\max}$ for all $0\le i\le j$.

\textbf{Base Step}: $i=1$. Consider an arbitrary $v\in\mathsf{B}'_{n-1}$. By $\mathsf{B}'_n\subseteq\mathsf{V}'_n$, we know that $[v']_\sim$ is final in $G$ for all $v'\in\mathrm{Post}_{n-1}(v)$. If $v\oplus(N_v\cdot\rho)\in\mathrm{Delay}_{n-1}(v)$, then $\mathsf{h}[v\oplus(N_v\cdot\rho)]\in\mathsf{B}_m^{\max}$ by Lemma~\ref{lem:region:auxi}. Therefore from $\Gamma''_m$ we have
\begin{eqnarray*}
1-\sum_{\mathsf{h}[v']\in\mathsf{B}_m}\mathbf{A}(\mathsf{h}[v],\mathsf{h}[v'])&\ge & p_{\min}\cdot{\left(\frac{1}{1+\rho\cdot\lambda(v)}\right)}^{N_v}\\
&\ge & p_{\min}\cdot{\left(\frac{1}{1+\rho\cdot\lambda(v)}\right)}^{T_{\max}\slash\rho}\\
&\ge & p_{\min}\cdot e^{-\lambda_{\max}T_{\max}}\\
&\ge & \mathfrak{c}
\end{eqnarray*}
Otherwise, from Lemma~\ref{lem:region:auxi} there are at least $\lfloor\delta\slash\rho\rfloor=k$ distinct elements in $\mathrm{Delay}_{n-1}(v)$. Note that $k\rho\ge\frac{1}{2}|V|^{-2}$. We have
\begin{eqnarray*}
& &1-\sum_{\mathsf{h}[v']\in\mathsf{B}_m}\mathbf{A}(\mathsf{h}[v],\mathsf{h}[v'])\\
&\ge&{\left(\frac{1}{1+\rho\cdot\lambda(v)}\right)}^{T_{\max}\slash\rho}\cdot
\frac{k\rho\cdot\lambda(v)}{1+\rho\cdot\lambda(v)}\cdot p_{\min}\\
&\ge& e^{-\lambda_{\max}T_{\max}}\cdot\frac{\frac{1}{2}|V|^{-2}\cdot\lambda(v)}{1+\frac{1}{2}|V|^{-2}\cdot\lambda(v)}\cdot p_{\min}\\
&\ge& e^{-\lambda_{\max}T_{\max}}\cdot p_{\min}\cdot \frac{\lambda_{\min}}{2|V|^2+\lambda_{\min}}
\end{eqnarray*}
Thus we have $|(\mathbf{A}\zeta_{\max})(v)|\le (1-\mathfrak{c})\cdot M_3\rho$.

\textbf{Inductive Step}: Suppose $(\mathbf{A}^i\zeta_{\max})(v)\le (1-\mathfrak{c}^i)\cdot M_3\rho$ for all $v\in \mathsf{B}'_{n-i}$. We prove the case for $i+1$. Fix some $v\in\mathsf{B}'_{n-(i+1)}$. If $v\oplus(N_v\cdot\rho)\in\mathrm{Delay}_{n-(i+1)}(v)$, then similar to the base step we have
\[
(\mathbf{A}^{i+1}\zeta_{\max})(v)\le\left(1-p_{\min}\cdot\left(\frac{1}{1+\rho\cdot\lambda(v)}\right)^{N_v}\cdot\mathfrak{c}^i\right)\cdot M_3\rho\le (1-\mathfrak{c}^{i+1})\cdot M_3\rho
\]
Otherwise, there are at least $\lfloor\delta\slash\rho\rfloor=k$ distinct elements in $\mathrm{Delay}_{n-(i+1)}(v)$.
By $\mathrm{Post}_{n-(i+1)}(v)\subseteq\mathsf{B}'_{n-i}$ and $\mathbf{A}^i\zeta_{\max}\le\zeta_{\max}$, we obtain that
\begin{eqnarray*}
&   & (\mathbf{A}^{i+1}\zeta_{\max})(v)\\
&\le& M_3\rho\cdot \left\{1-{\left(\frac{1}{1+\rho\cdot\lambda(v)}\right)}^{N_v}\cdot\frac{k\rho\cdot\lambda(v)}{1+\rho\cdot\lambda(v)}\cdot p_{\min}\cdot\mathfrak{c}^i\right\}\\
&\le& M_3\rho\cdot \left\{1-e^{-\lambda_{\max}T_{\max}}\cdot p_{\min}\cdot \frac{\lambda_{\min}}{2|V|^2+\lambda_{\min}}\cdot \mathfrak{c}^i\right\}\\
&=& M_3\rho\cdot (1-\mathfrak{c}^{i+1})
\end{eqnarray*}
Then the result follows. Then we obtain that
\[
(\mathbf{A}^{|V|-1}\zeta_{\max})(v)\le (\mathbf{A}^i\zeta_{\max})(v)\le (1-\mathfrak{c}^i)\cdot M_3\rho\le (1-\mathfrak{c}^{|V|-1})\cdot M_3\rho
\]
for all $1\le i\le n-1$ and $v\in\mathsf{B}'_{n-i}$. Thus $\mathbf{A}^{|V|-1}\zeta_{\max}(v^*)\le (1-\mathfrak{c}^{|V|-1})\cdot M_3\rho$ for all $v^*\in\mathsf{B}_m^{\max}$.

Now consider an arbitrary $v\in\mathsf{B}_m$ while $\zeta=\zeta_{\max}$. If either $v\oplus (N_v\rho)\not\in\mathsf{B}_m^{\max}$ or $\mathsf{h}[(v\oplus (N_v\rho))^+_u]\not\in\mathsf{B}_m$ for some $u\in S$ such that $\mathbf{P}(v,u)>0$, then we have
\begin{eqnarray*}
(\mathbf{A}^{|V|}\zeta_{\max})(v)&\le &(\mathbf{A}\zeta_{\max})(v)\\
&\le &\left(1-\left(\frac{1}{1+\rho\lambda(v)}\right)^{T_{\max}\slash\rho}\cdot p_{\min}\right)\cdot M_3\rho\\
&\le &(1-e^{-\lambda_{\max}T_{\max}}\cdot p_{\min})\cdot M_3\rho\\
&\le &(1-\mathfrak{c}^{|V|})\cdot M_3\rho
\end{eqnarray*}
Otherwise, by
\[
\sum_{u\in S}\mathbf{P}(v^*,u)\cdot(\mathbf{A}^{|V|-1}\zeta_{\max})((v^*)^+_u)=(\mathbf{A}^{|V|}\zeta_{\max})(v^*)\le (\mathbf{A}^{|V|-1}\zeta_{\max})(v^*)
\]
where $v^*=v\oplus (N_v\rho)$, we have
\begin{eqnarray*}
(\mathbf{A}\cdot\mathbf{A}^{|V|-1}\zeta_{\max})(v)&\le &\left(1-\left(\frac{1}{1+\rho\lambda(v)}\right)^{T_{\max}\slash\rho}\cdot\mathfrak{c}^{|V|-1}\right)\cdot M_3\rho\\
&\le & (1-\mathfrak{c}^{|V|})\cdot M_3\rho
\end{eqnarray*}
Then we have $\mathbf{A}^{|V|}\zeta_{\max}\le (1-\mathfrak{c}^{|V|})\cdot\zeta_{\max}$. It follows that
\[
\mathbf{A}^{i|V|}\zeta_{\max}\le (1-\mathfrak{c}^{|V|})^i\zeta_{\max}
\]
for all $i\in\mathbb{N}$. Thus by the monotonicity of $\mathbf{A}$, we have $\sum_{i=0}^\infty\mathbf{A}^i\zeta_{\max}$ converges since $\sum_{i=0}^\infty\mathbf{A}^i\zeta_{\max}$ is bounded by $|V|\cdot\mathfrak{c}^{-|V|}\cdot\zeta_{\max}$.

Finally we consider any $\zeta$ such that $|\zeta|\le \zeta_{\max}$. Note that $\left|\mathbf{A}^i\zeta\right|\le \mathbf{A}^i\zeta_{\max}$ since all entries of $\mathbf{A}^i$ are non-negative. Thus by Cauchy Criterion, it follows that $\sum_{i=0}^{\infty}\mathbf{A}^i\zeta$ converges and $\parallel\sum_{i=0}^{\infty}\mathbf{A}^i\zeta\parallel_\infty\le |V|\cdot\mathfrak{c}^{-|V|}\cdot M_3\rho$.
\end{proof}
By Proposition~\ref{lem:leq:existence}, the system of linear equation $\mu=\mathbf{A}\mu+\zeta$ has a solution for all $|\zeta|\le\zeta_{\max}$ when $m>2|V|^2$. Below we assume that $m>2|V|^2$. The following lemmas show that the linear equation has a unique solution.

\begin{proposition}\label{lem:leq:bound}
For all solutions $\mu$ of $\mu=\mathbf{A}\mu+\zeta$ with ${\parallel\zeta\parallel}_\infty<M_3\rho$, we have $|\mu|\le\mu^*$, where $\mu^*:=\sum_{i=0}^\infty\mathbf{A}^i\zeta_{\max}$.
\end{proposition}
\begin{proof}
Let $\mu$ be an arbitrary solution of $\mu=\mathbf{A}\mu+\zeta$. Define $\mu'=\mu^*-\mu$. By the fact that $\mu^*=\mathbf{A}\mu^*+\zeta_{\max}$, we have $\mu'(\mathsf{h}[v])>(\mathbf{A}\mu')(\mathsf{h}[v])$ for all $\mathsf{h}[v]\in\mathsf{B}_m$. Suppose there is some $\mathsf{h}[v]\in\mathsf{B}_m$ such that $\mu'(\mathsf{h}[v])<0$. W.l.o.g we assume that $\mu'(\mathsf{h}[v])$ is the least element of $\{\mu'(\mathsf{h}[v'])\mid \mathsf{h}[v']\in\mathsf{B}_m\}$. Denote $c:=\sum_{\mathsf{h}[v']\in\mathsf{B}_m}\mathbf{A}(\mathsf{h}[v],\mathsf{h}[v'])\in [0,1]$. We have $\mu'(\mathsf{h}[v])>c\cdot\mu'(\mathsf{h}[v])$, which implies $c>1$. Contradiction. Thus $\mu'\ge\vec{0}$. Similar arguments holds if we define $\mu'=\mu^*+\mu$. Thus we have $|\mu|\le\mu^*$.
\end{proof}
\begin{proposition}
The system of linear equations $\mu=\mathbf{A}\mu+\zeta$ has a unique solution for all $\zeta$ such that ${\parallel}{\zeta}{\parallel}_\infty<M_3\rho$. It follows that $\mathbf{I}-\mathbf{A}$ is invertible where $\mathbf{I}$ is the identity matrix.
\end{proposition}
\begin{proof}
By Lemma~\ref{lem:leq:existence}, the system $\mu=\mathbf{A}\mu+\zeta$ has a solution. And by Lemma~\ref{lem:leq:bound}, all solutions of $\mu=\mathbf{A}\mu+\zeta$ are bounded by $\mu^*$. Suppose it has two distinct solutions. Then the homogeneous system of linear equations $\mu=\mathbf{A}\mu$ has a non-trivial solution, which implies that the solutions of $\mu=\mathbf{A}\mu+\zeta$ cannot be bounded. Contradiction. Thus $\mu=\mathbf{A}\mu+\zeta$ has a unique solution and $\mathbf{I}-\mathbf{A}$ is invertible.
\end{proof}
Now we analyse $\Gamma'_m$ ($\Gamma_m$). In the following theorem which is the main result of the paper, we prove that the equation $\mu=\mathbf{C}\mu+\mathbf{d}$ has a unique solution (i.e. $\mathbf{I}-\mathbf{C}$ is invertible), and give the error bound between the unique solution and the function $\mathsf{prob}$.

\begin{theorem}\label{thm:main}
The matrix equation $\mu=\mathbf{C}\mu+\mathbf{d}$ (for $\Gamma'_m$) has a unique solution $\overline{\mu}$. Moreover, $\max_{\mathsf{h}[v]\in\mathsf{B}_m}\left|\overline{\mu}(\mathsf{h}[v])-\mathsf{prob}(v)\right|\le |V|\cdot{\mathfrak{c}}^{-|V|}\cdot M_3\rho$.
\end{theorem}
\begin{proof}
We first prove that $\mu=\mathbf{C}\mu+\mathbf{d}$ has a unique solution. Let $\mu=\mathbf{C}\mu+\zeta$ be a matrix equation such that ${\parallel\zeta\parallel}_\infty< M_2\rho^2$. From the proof of Proposition~\ref{lem:approx:equiv}, we can equivalently expand this equation into some equation $\mu=\mathbf{A}\mu+\zeta'$ with ${\parallel\zeta'\parallel}_\infty<T_{\max}\slash\rho\cdot M_2\rho^2=M_3\cdot\rho$. Since $\mu=\mathbf{A}\mu+\zeta'$ has a unique solution, we have $\mu=\mathbf{C}\mu+\zeta$ also has a unique solution. Thus $\mathbf{I}-\mathbf{C}$ is invertible and $\mu=\mathbf{C}\mu+\mathbf{d}$ has a unique solution.

Now we prove the error bound between $\overline{\mu}$ and $\mathsf{prob}$. Define the vector $\mu'$ such that $\mu'(\mathsf{h}[v])=\overline{\mu}(\mathsf{h}[v])-\mathsf{prob}(v)$ for all $\mathsf{h}[v]\in\mathsf{B}_m$. By Proposition~\ref{lem:err:gamma1}, $\mu'$ is the unique solution of $\mu=\mathbf{C}\mu+\zeta$ for some $\parallel\zeta\parallel_\infty<M_2\rho^2$. Then $\mu'$ is also the unique solution of the equation $\mu=\mathbf{A}\mu+\zeta'$ for some $\parallel\zeta'\parallel_\infty<M_3\rho$. By Proposition~\ref{lem:leq:existence},  $\parallel\mu'\parallel_\infty\le|V|\cdot{\mathfrak{c}}^{-|V|}\cdot M_3\rho$.
\end{proof}
By Theorem~\ref{thm:main} and the Lipschitz Continuity (Corollary~\ref{crlly:lipschitz}), we can approximate $\mathsf{prob}(s,q,\eta)$ as follows: given $\epsilon\in (0,1)$, we choose $m$ sufficiently large and some $\mathsf{h}[v]\in\mathsf{D}_m$ such that $|\mathsf{prob}(v)-\mathsf{prob}(s,q,\eta)|<\frac{1}{2}\epsilon$ and $M_3|V|\mathfrak{c}^{-|V|}\cdot\rho<\frac{1}{2}\epsilon$. Then we solve the system $\Gamma'_m$ to obtain $\overline{\mu}(\mathsf{h}[v])$.

\section{Conclusion and Future Work}

We have shown an algorithm to approximate the acceptance probabilities of CTMC-paths by a multi-clock DTA under finite acceptance condition. Unlike the result by Barbot~\emph{et al.}~\cite{DBLP:conf/tacas/BarbotCHKM11}, we are able to derive an approximation error. Chen~\emph{et al.}~\cite{DBLP:journals/corr/abs-1101-3694} demonstrated that computing the acceptance probability of CTMC-paths by a multi-clock DTA under Muller acceptance condition can be reduced to the one under finite acceptance condition. Thus our result can also be applied to Muller acceptance conditions. One future direction is to refine our approximation algorithm by importing zone-based techniques~\cite{DBLP:conf/ac/BengtssonY03}. Another future direction is to extend this result to continuous-time Markov decision processes (CTMDP)~\cite{DBLP:conf/formats/BertrandS12} or continuous-time Markov games (CTMG)~\cite{DBLP:conf/concur/BrazdilKKKR10,DBLP:conf/fsttcs/FearnleyRSZ11}. A more challenging task would be to consider the acceptance probabilities of CTMC-paths by a non-deterministic timed automaton.

\subsubsection*{Acknowledgement}

I thank Prof. Joost-Pieter Katoen for valuable comments on the writing of this paper. The author is supported by a CSC (China Scholarship Council) scholarship.

\bibliographystyle{plain}
\bibliography{references}

\begin{thebibliography}{10}

\bibitem{DBLP:journals/tcs/AlurD94}
Rajeev Alur and David~L. Dill.
\newblock A theory of timed automata.
\newblock {\em Theor. Comput. Sci.}, 126(2):183--235, 1994.

\bibitem{DBLP:journals/tse/BaierCHKS07}
Christel Baier, Lucia Cloth, Boudewijn~R. Haverkort, Matthias Kuntz, and Markus
  Siegle.
\newblock Model checking {M}arkov chains with actions and state labels.
\newblock {\em IEEE Trans. Software Eng.}, 33(4):209--224, 2007.

\bibitem{DBLP:journals/tse/BaierHHK03}
Christel Baier, Boudewijn~R. Haverkort, Holger Hermanns, and Joost-Pieter
  Katoen.
\newblock Model-checking algorithms for continuous-time {M}arkov chains.
\newblock {\em IEEE Trans. Software Eng.}, 29(6):524--541, 2003.

\bibitem{DBLP:journals/cacm/BaierHHK10}
Christel Baier, Boudewijn~R. Haverkort, Holger Hermanns, and Joost-Pieter
  Katoen.
\newblock Performance evaluation and model checking join forces.
\newblock {\em Commun. ACM}, 53(9):76--85, 2010.

\bibitem{DBLP:books/daglib/0020348}
Christel Baier and Joost-Pieter Katoen.
\newblock {\em {P}rinciples of {M}odel {C}hecking}.
\newblock MIT Press, 2008.

\bibitem{DBLP:conf/tacas/BarbotCHKM11}
Beno\^{\i}t Barbot, Taolue Chen, Tingting Han, Joost-Pieter Katoen, and
  Alexandru Mereacre.
\newblock Efficient {CTMC} model checking of linear real-time objectives.
\newblock In Parosh~Aziz Abdulla and K.~Rustan~M. Leino, editors, {\em TACAS},
  volume 6605 of {\em Lecture Notes in Computer Science}, pages 128--142.
  Springer, 2011.

\bibitem{DBLP:conf/ac/BengtssonY03}
Johan Bengtsson and Wang Yi.
\newblock Timed automata: Semantics, algorithms and tools.
\newblock In J{\"o}rg Desel, Wolfgang Reisig, and Grzegorz Rozenberg, editors,
  {\em Lectures on Concurrency and Petri Nets}, volume 3098 of {\em Lecture
  Notes in Computer Science}, pages 87--124. Springer, 2003.

\bibitem{DBLP:conf/formats/BertrandS12}
Nathalie Bertrand and Sven Schewe.
\newblock Playing optimally on timed automata with random delays.
\newblock In Marcin Jurdzinski and Dejan Nickovic, editors, {\em FORMATS},
  volume 7595 of {\em Lecture Notes in Computer Science}, pages 43--58.
  Springer, 2012.

\bibitem{patrick.billingsley.pm}
Patrick Billingsley.
\newblock {\em {Probability and Measure}}.
\newblock John Wiley \& Sons, New York, NY, USA, 1995.

\bibitem{DBLP:conf/concur/BrazdilKKKR10}
Tom{\'a}s Br{\'a}zdil, Jan Krc{\'a}l, Jan Kret\'{\i}nsk{\'y}, Anton\'{\i}n
  Ku\v{c}era, and Vojtech Reh{\'a}k.
\newblock Stochastic real-time games with qualitative timed automata
  objectives.
\newblock In Paul Gastin and Fran\c{c}ois Laroussinie, editors, {\em CONCUR},
  volume 6269 of {\em Lecture Notes in Computer Science}, pages 207--221.
  Springer, 2010.

\bibitem{DBLP:conf/hybrid/BrazdilKKKR11}
Tom{\'a}s Br{\'a}zdil, Jan Krc{\'a}l, Jan Kret\'{\i}nsk{\'y}, Anton\'{\i}n
  Ku\v{c}era, and Vojtech Reh{\'a}k.
\newblock Measuring performance of continuous-time stochastic processes using
  timed automata.
\newblock In Marco Caccamo, Emilio Frazzoli, and Radu Grosu, editors, {\em
  HSCC}, pages 33--42. ACM, 2011.

\bibitem{DBLP:conf/formats/ChenDKM11}
Taolue Chen, Marco Diciolla, Marta~Z. Kwiatkowska, and Alexandru Mereacre.
\newblock Time-bounded verification of {CTMCs} against real-time
  specifications.
\newblock In Uli Fahrenberg and Stavros Tripakis, editors, {\em FORMATS},
  volume 6919 of {\em Lecture Notes in Computer Science}, pages 26--42.
  Springer, 2011.

\bibitem{DBLP:journals/corr/abs-1101-3694}
Taolue Chen, Tingting Han, Joost-Pieter Katoen, and Alexandru Mereacre.
\newblock {Model Checking of Continuous-Time Markov Chains Against Timed
  Automata Specifications}.
\newblock {\em Logical Methods in Computer Science}, 7(1), 2011.

\bibitem{mha.davis}
M.H.A Davis.
\newblock {\em {M}arkov {M}odels and {O}ptimizations}.
\newblock Chapman \& Hall, New York, NY, USA, 1993.

\bibitem{DBLP:journals/tse/DonatelliHS09}
Susanna Donatelli, Serge Haddad, and Jeremy Sproston.
\newblock Model checking timed and stochastic properties with ${CSL}^{{TA}}$.
\newblock {\em IEEE Trans. Software Eng.}, 35(2):224--240, 2009.

\bibitem{DBLP:conf/fsttcs/FearnleyRSZ11}
John Fearnley, Markus Rabe, Sven Schewe, and Lijun Zhang.
\newblock Efficient approximation of optimal control for continuous-time
  {M}arkov games.
\newblock In Supratik Chakraborty and Amit Kumar, editors, {\em FSTTCS},
  volume~13 of {\em LIPIcs}, pages 399--410. Schloss Dagstuhl - Leibniz-Zentrum
  fuer Informatik, 2011.

\bibitem{feller.william.aiptia}
William Feller.
\newblock {\em {An Introduction to Probability Theory and Its Applications}}.
\newblock John Wiley \& Sons, New York, NY, USA, 1966.

\bibitem{mikeev.neuhausser.spieler.wolf}
Linar Mikeev, Martin~R. Neuh\"{a}u\ss{}er, David Spieler, and Verena Wolf.
\newblock On-the-fly verification and optimization of {DTA}-properties for
  large {M}arkov chains.
\newblock {\em Form. Methods Syst. Des.}, 2012.

\bibitem{Rudin:1987:RCA:26851}
Walter Rudin.
\newblock {\em {R}eal and {C}omplex {A}nalysis, 3rd ed.}
\newblock McGraw-Hill, Inc., New York, NY, USA, 1987.

\bibitem{jw.thomas}
J.W. Thomas.
\newblock {\em {N}umerical {P}artial {D}ifferential {E}quations: {F}inite
  {D}ifference {M}ethods}.
\newblock Springer-Verlag, New York, NY, USA, 1995.

\bibitem{DBLP:conf/icalp/ZhangJNH11}
Lijun Zhang, David~N. Jansen, Flemming Nielson, and Holger Hermanns.
\newblock Automata-based {CSL} model checking.
\newblock In Luca Aceto, Monika Henzinger, and Jiri Sgall, editors, {\em ICALP
  (2)}, volume 6756 of {\em Lecture Notes in Computer Science}, pages 271--282.
  Springer, 2011.

\end{thebibliography}

\end{document}